\documentclass[12pt]{article}
\usepackage[margin=1in]{geometry}
\usepackage{amsmath}
\usepackage{graphicx}
\usepackage{enumerate}

\usepackage{url} 

\newcommand{\blind}{1}

\usepackage{amsmath,amsthm, amssymb, bbm}
\usepackage{blindtext}
\usepackage{enumerate}
\usepackage{listings}
\usepackage{color, colortbl}
\usepackage{xcolor}
\usepackage[nottoc]{tocbibind}
\usepackage[utf8]{inputenc}
\usepackage{caption}
\usepackage{tabularx}
\usepackage{array}
\usepackage{rotating}
\usepackage{wrapfig}
\usepackage{graphicx}
\usepackage{subfig}
\definecolor {processblue}{cmyk}{0.96,0,0,0}
\usepackage{mathrsfs}
\usepackage{hyperref}
\hypersetup{
colorlinks   = true,
linkcolor = red,
citecolor    = blue
}
\usepackage[square]{natbib} 
\usepackage{fancyvrb}
\usepackage{enumitem}
\definecolor{LightCyan}{rgb}{0.88,1,1}
\usepackage{parskip}
\usepackage{float}
\usepackage{titling}
\usepackage{varwidth}
\usepackage{multirow}
\usepackage{breqn}
\usepackage{bm}
\usepackage{titlesec}
\titleformat{\subsection}    
{\bfseries\itshape}{\thesubsection}{1em}{}
\usepackage{xr}
\usepackage{setspace}

\newtheorem{theorem}{Theorem}
\newtheorem{lemma}{Lemma}

\theoremstyle{remark}

\newtheorem{assumption}{Assumption}
\newtheorem{remark}{Remark}

\makeatletter
\newcommand*{\addFileDependency}[1]{
\typeout{(#1)}
\@addtofilelist{#1}
\IfFileExists{#1}{}{\typeout{No file #1.}}
}
\makeatother

\def\inprob{\stackrel{p}{\rightarrow}}
\def\indist{\rightsquigarrow}
\def\ind{\perp\!\!\!\perp}
\def\notind{\not\!\perp\!\!\!\perp}

\newcommand{\var}{\text{var}}

\newcommand{\Pb}{\mathbb{P}}
\newcommand{\Qb}{\mathbb{Q}}
\newcommand{\Pn}{\mathbb{P}_n}
\newcommand{\Gn}{\mathbb{G}_n}
\newcommand{\Gb}{\mathbb{G}}
\newcommand{\E}{\mathbb{E}}
\newcommand{\R}{\mathbb{R}}
\newcommand{\Ob}{\mathbf{O}}
\newcommand{\Xb}{\mathbf{X}}
\newcommand{\xb}{\mathbf{x}}
\newcommand{\Ec}{\mathcal{E}}

\newcommand{\etab}{\pmb{\eta}}

\DeclareMathOperator*{\argmin}{arg\,min}

\DeclareMathOperator{\sgn}{sgn}
\newcommand{\one}{\mathbbm{1}}

\begin{document}

\def\spacingset#1{\renewcommand{\baselinestretch}%
	{#1}\small\normalsize} \spacingset{1}


\if1\blind
{
	\title{\bf Sensitivity Analysis via the Proportion of Unmeasured Confounding}
	\author{Matteo Bonvini\thanks{Department of Statistics \& Data Science, Carnegie Mellon University, 5000 Forbes Avenue, Pittsburgh, PA 15213. Email: mbonvini@stat.cmu.edu} \and Edward H. Kennedy\thanks{Assistant Professor, Department of Statistics \& Data Science, Carnegie Mellon University, 5000 Forbes Avenue, Pittsburgh, PA 15213. Email: edward@stat.cmu.edu}}
	\maketitle
} \fi

\if0\blind
{
	\bigskip
	\bigskip
	\bigskip
	\begin{center}
		{\LARGE\bf Title}
	\end{center}
	\medskip
} \fi

\bigskip
\begin{abstract}
	In observational studies, identification of ATEs is generally achieved by assuming that the correct set of confounders has been measured and properly included in the relevant models. Because this assumption is both strong and untestable, a sensitivity analysis should be performed. Common approaches include modeling the bias directly or varying the propensity scores to probe the effects of a potential unmeasured confounder. In this paper, we take a novel approach whereby the sensitivity parameter is the ``proportion of unmeasured confounding:" the proportion of units for whom the treatment is not as good as randomized even after conditioning on the observed covariates. We consider different assumptions on the probability of a unit being unconfounded. In each case, we derive sharp bounds on the average treatment effect as a function of the sensitivity parameter and propose nonparametric estimators that allow flexible covariate adjustment. We also introduce a one-number summary of a study's robustness to the number of confounded units. Finally, we explore finite-sample properties via simulation, and apply the methods to an observational database used to assess the effects of right heart catheterization. 
\end{abstract}

\noindent%
\vfill

\newpage
\section{Introduction}
In an experiment, the random assignment of the treatment to the units ensures that any measured and unmeasured factors are balanced between the treatment and control groups, thereby allowing the researcher to attribute any observed effect to the treatment. In observational studies, however, achieving such balance requires the untestable assumption that all confounders, roughly variables affecting both the treatment $A$ and the outcome $Y$, are collected. To gauge the consequences of departures from the no-unmeasured-confounding assumption, a sensitivity analysis generally posits the existence of an unmeasured confounder $U$ and varies either the $U$-$A$ association or the $U$-$Y$ association or both. The minimal strength of these associations that would drive the observed $Y$-$A$ association to zero is often reported as a measure of the study's robustness to unmeasured confounding. 

Since the seminal work of \cite{cornfield1959smoking} on the association between smoking and lung cancer, a plethora of sensitivity analysis frameworks have been proposed. Here, we mention a few of them and refer to \cite{liu2013introduction} and \cite{richardson2014nonparametric} for excellent reviews. In the context of matched studies, Rosenbaum's framework \citep{rosenbaum1987sensitivity, rosenbaum2002covariance} is likely the most commonly used. It governs the $U$-$A$ association via a parameter $\Gamma \geq 1$ by requiring that, within each pair, the ratio of the odds that unit 1 is treated to the odds that unit 2 is treated falls in the interval $[\Gamma^{-1}, \Gamma]$. The $U$-$Y$ association is often left unrestricted or bounded as in \cite{gastwirth1998dual}. More recently, \cite{zhao2017sensitivity} and \cite{yadlowsky2018bounds} have proposed extensions to this framework that do not require matching. 

In addition, \cite{ding2016sensitivity} and \cite{vanderweele2017sensitivity} have derived a bounding factor for certain treatment effects in terms of two sensitivity parameters governing the $U$-$A$ and $U$-$Y$ relationships. Other authors have proposed modeling the distribution of $U$ and the relationships $U-Y$ and $U-A$ directly \citep{imbens2003sensitivity, rosenbaum1983assessing}, which has been recently extended to the case where the distribution of $U$ is left unspecified by \cite{zhang2019sens}. In the context of time-varying treatments, sensitivity analyses have been proposed for marginal structural models \citep{brumback2004sensitivity} and cause-specific selection models \citep{andrea2001methods}. 

In this paper, we propose a novel approach to sensitivity analysis based on a mixture model for confounding. We conceptualize that an unknown fraction $\epsilon$ of the units in the sample is arbitrarily confounded while the rest is not. The parameter $\epsilon$ is unknown and not estimable but can be varied as a sensitivity parameter. As discussed below, our model generalizes some relaxations to the no-unmeasured-confounding assumption that have been previously proposed in the literature. Furthermore, our framework yields a natural one-number summary of a study's robustness: the minimum proportion of confounded units such that bounds on the average treatment effect contain zero. All the code can be found in the Github repository \texttt{matteobonvini/experiments-sensitivity-paper}. 

\subsection{Motivation}\label{sec:motivation}
The most widely adopted frameworks for sensitivity analysis generally assume that each unit in the sample could be subject to unmeasured confounding and then proceed by specifying the maximal extent of such confounding. However, just like a treatment effect can be heterogeneous, confounding, too, can differ between units. We propose a complementary approach: in some instances, the researcher may have failed to measure relevant confounders but may hope that there is a subset of units, possibly unknown, for whom the treatment is as good as randomized given the measured covariates.

As a toy example, suppose it is observed that adolescent alcohol drinking (treatment $A$) is positively associated with the occurrence of liver diseases (outcome $Y$). Suppose all confounders $X$ have been recorded except for parental smoking, which could be associated with both $A$ \citep{oliveira2019influence, pengpid2019alcohol} and $Y$ due to second-hand smoking \citep{lammert2013questionnaire}. Previously proposed sensitivity analyses would check whether a small association between parental smoking and $A$ or $Y$ can explain away the observed A-Y association. Instead, we propose to leverage on the observation that parental smoking is a confounder only for units whose parents smoke at home. For instance, some parents may only smoke at work, in which case parental smoking would not have an effect on $Y$. The sample is thus composed of two groups: those units for which $A$ is as good as randomized given $X$ because they are not subject to second-hand smoking regardless of whether their parents smoke and those for which it is not. Depending on how prevalent the former group is, the observed $A$--$Y$ association might be at least partially attributed to the effect of $A$. This toy example generalizes to other cases. For instance, if a confounder is measured with error, the observed covariates may be sufficient to de-confound the treatment-outcome relationship only for an unknown subset of units. In such case, the sample can be thought of containing two groups: those units for whom the confounder was measured correctly, e.g. if the questionnaire on motivation or drugs usage was answered truthfully, and those for whom it was not.

The possibility that a sample comes from a mixture of distributions has been studied in great detail in statistics. In robust statistics, for example, it is assumed that a small unknown fraction of the sample comes from a ``corrupted" or ``contaminated" distribution that is not the target of inference (see Remark \ref{remark_robust_stat}). In causal inference, unmeasured confounding takes the role of contamination. Borrowing the contaminated model from this literature, we conceptualize that an unknown fraction of the sample suffers from unmeasured confounding. 

For example, consider Figure \ref{ex_partition}. In the shaded region of the space defined by the two observed covariates, the treatment is not assigned randomly; units with covariates' values falling in this region may have self-selected into the treatment arms and therefore estimating the effect of the treatment on their outcomes is impossible without making further, untestable assumptions. For brevity, we say these units are ``confounded," while the other units are ``unconfounded." Note that, except in special cases, some of which are discussed next, the region is not identifiable from the observed data. However, even if the region is not identifiable, its measure, termed $\epsilon$ in our model, might be specified or upper bounded using subject-matter knowledge. More generally, $\epsilon$ can be varied as a sensitivity parameter. In Figure \ref{ex_partition}, despite covering different sets of units, all three regions have the same mass, with approximately 20\% of the points falling inside them. Given a value for $\epsilon$, we show how to find the region yielding the most conservative inference.  

\begin{figure*}[h!]
	\subfloat{%
		\includegraphics[width=0.3\textwidth]{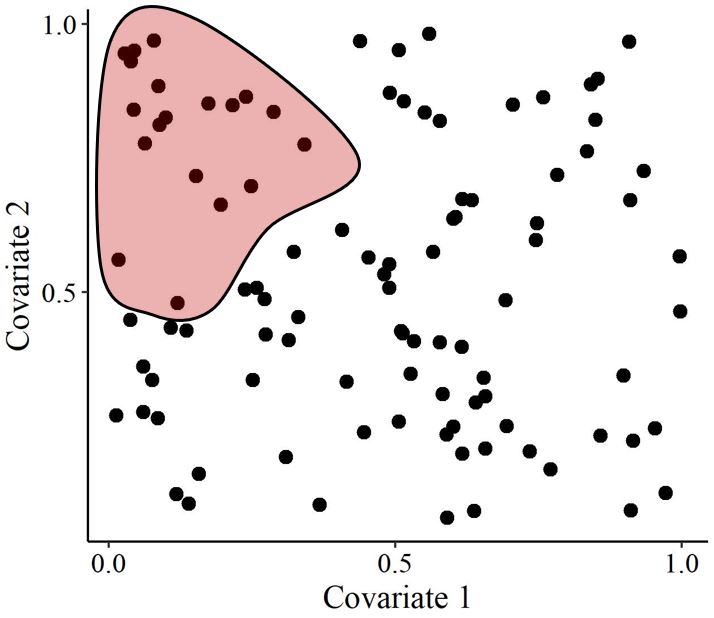}}
	\hspace{\fill}
	\subfloat{%
		\includegraphics[width=0.3\textwidth]{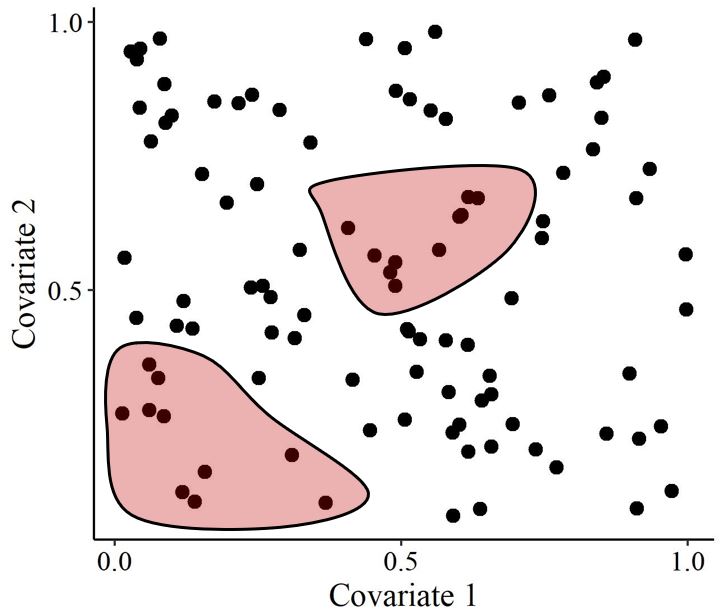}}
	\hspace{\fill}
	\subfloat{%
		\includegraphics[width=0.3\textwidth]{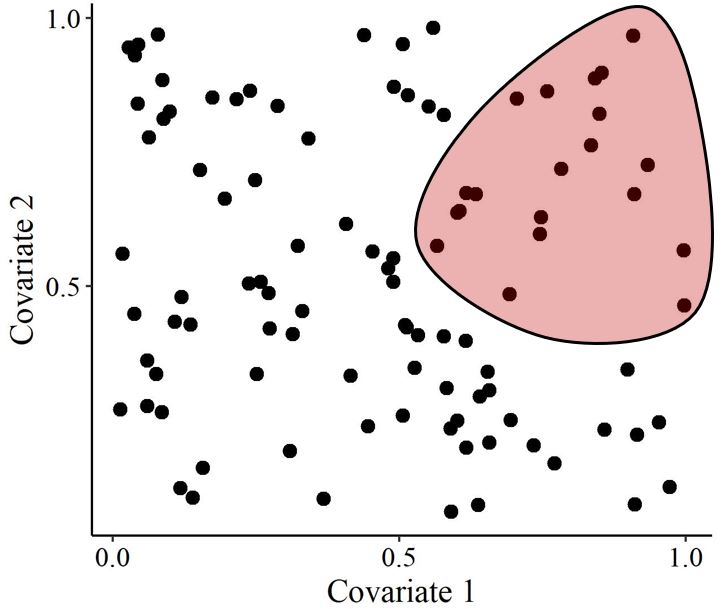}}\\
	\caption{\label{ex_partition} The shaded region represents the set of units for whom the treatment is not assigned randomly, even after conditioning on observed covariates. All three figures show approximately the same number of points falling within the ``confounded region," albeit covering different sets of units. The probability $\epsilon$ that a unit falls within the region is our model's sensitivity parameter, here $\epsilon \approx 0.2$.}
\end{figure*}

Special cases of our model have already been discussed in the literature when it is known who the confounded units are. For example, in introducing the selective ignorability framework, \cite{joffe2010selective} discuss estimating the effect of erythropoietin alpha (EPO) on mortality using an observational database containing information on all subjects in the United States on hemodialysis. The treatment is thought to be unconfounded only after conditioning on hematocrit, which, however, is not recorded for 10.6\% of the subjects. Thus, one may view 10.6\% of the sample as coming from a ``confounded distribution." In addition, the differential effects framework proposed in \cite{rosenbaum2006differential}, too, can be regarded as a special case of our model. Differential effects are treatment contrasts that are immune to certain types of biases called ``generic biases." For example, suppose two treatments are under study. In certain cases, it is plausible that, while units might self select into either treatment arm, the choice of the treatment among units who take exactly one treatment is as good as random. Notice that this setup is a special case of our model: the confounded units are precisely those who are not taking any treatment or are taking a combination of both of them.

Finally, a standard instrumental variables (IV) setting, too, can be thought of as a case where a fraction of the units is unconfounded. For example, consider an experiment with binary treatment that suffers from units' non-compliance. The treatment assignment is randomized but the treatment received is not. For the units who complied with the experimental guidelines, the treatment received is equal to the treatment assigned, which is randomly assigned. Thus, the compliers can be considered the units for whom the treatment / outcome relationship is not confounded. In fact, in their detailed analysis of the binary IV model, \cite{richardson2010analysis} propose a sensitivity analysis for the average treatment effect where the sensitivity parameter can be expressed as the proportion of compliers. For the observational setting considered in this paper, however, the instrument is never observed, thus, contrarily to a standard IV analysis, the sample contains no information regarding who the confounded units are. In this light, our contribution can also be regarded as an attempt to infer average treatment effects when it is plausible that nature is acting via an unobservable IV. 
\section{The Sensitivity Model}\label{section:model}
We suppose we are given an iid sample $(\Ob_1, \ldots, \Ob_n) \sim \Pb$ with $\Ob = (\Xb, A, Y)$, for covariates $\Xb \in \mathcal{X} \subseteq \R^p$, a binary treatment $A \in \{0, 1\}$ and an outcome $Y \in \mathcal{Y} \subseteq \R $. We let $Y^a$ denote the potential outcome that would have been observed had the treatment been set to $A = a$ \citep{rubin1974estimating}. The goal is to estimate the Average Treatment Effect (ATE) defined as $\psi ~= \E(Y^1 - Y^0)$. To ease the notation, we let $\pi(a\mid\Xb) = \Pb(A = a \mid \Xb)$, 
\begin{align*}
\mu_a(\Xb) = \E\left(Y\mid A = a, \Xb\right), \quad \text{and} \quad \etab = \left\{ \pi(0 \mid \Xb), \pi(1 \mid \Xb), \mu_0(\Xb), \mu_1(\Xb)  \right\}.
\end{align*}
Throughout, we assume that the following two assumptions hold
\medskip
\begin{assumption}[Consistency] \label{assumption:consistency}
	$Y = A Y^1 + (1 - A) Y^0$.
\end{assumption}
\begin{assumption}[Positivity] \label{assumption:positivity} $\Pb\left\{t \leq \pi(a \mid \Xb) \leq 1 - t\right\} = 1$ for some $t > 0$.
\end{assumption}
Both assumptions are standard in the causal inference literature. Consistency rules out any interference between the units, whereas positivity requires that each unit has a non-zero chance of receiving either treatment arm regardless of their covariates' values. It is well known that if, in addition to consistency and positivity, it also holds that $Y^a \ind A \mid \Xb$ (no unmeasured confounding), then $\psi$ can be point-identified as $\psi = \E\{ \mu_1(\Xb) - \mu_0(\Xb) \}$. In this work, we propose a sensitivity model that relaxes the no-unmeasured-confounding assumption while retaining both consistency and positivity. As a consequence of this relaxation, $\psi$ is no longer point-identified but it can still be bounded. 

Our model supposes that the observed distribution $\Pb$ is derived from a counterfactual distribution $\Qb$ of $(\Xb, A, Y^1, Y^0)$ such that
\begin{align}\label{eq:counterfactual_contamination}
\Qb = \epsilon \Qb_0 + (1 - \epsilon) \Qb_1
\end{align}
where $\Qb_0$ is a ``confounded distribution" for which $A \notind Y^a \mid \Xb$ and $\Qb_1$ is an ``unconfounded distribution" for which $A \ind Y^a \mid \Xb$. In practice, it might be useful to think of each $\Qb_i$ as potentially factoring according to $A \ind Y^a \mid \mathcal{S}_i$, where $\mathcal{S}_i$ is a set of confounding variables such that $\mathcal{S}_1$ is measured but $\mathcal{S}_0 \setminus \mathcal{S}_1 \neq \emptyset$ is not. \footnote{As pointed out by an anonymous reviewer, the mixture model \eqref{eq:counterfactual_contamination} could be generalized to $\Qb~=~ \sum_{j = 1}^J\epsilon_j\Qb_j$, where each $Q_j$ is a distribution on the counterfactuals capturing different degrees of the confounding. While richer sensitivity analyses can yield more nuanced conclusions, the large number of parameters whose plausibility range would need to be assessed ($J - 1$ in this case) may hinder their applications in many settings.}\footnote{For instance, consider the toy example above, with $X = \emptyset$ and $Y, A, U \in \{0, 1\}$ for simplicity. Suppose that $\Pb(U = 1 \mid A) = \gamma_0 + \gamma_1A$ and $\Qb_s(Y^a = 1 \mid A, U) = \alpha_1 s + (1-s)(\alpha_2 + \alpha_3U)$, for some constants $\gamma$ and $\alpha$. Then, $\E_{\Qb_0}(Y^1 - Y^0) = \E_{\Qb_1}(Y^1 - Y^0) = 0$ and $\E_{\Qb_1}(Y^1 \mid A = 1) - \E_{\Qb_1}(Y^0 \mid A = 0) = 0$, but $\E_{\Qb_0}(Y^1 \mid A = 1) - \E_{\Qb_1}(Y^0 \mid A = 0) = \alpha_3\gamma_1$, which is generally nonzero.} 

The parameter $\epsilon \in \Ec \subseteq[0, 1]$ governs the proportion of unmeasured confounding. It is unknown and not estimable but can be varied as a sensitivity parameter. Here, $\Ec$ is an interval that the user can specify. Although $\psi$ cannot be point-identified for $\epsilon > 0$, it is possible to bound it as a function of $\epsilon$. In particular, for $\epsilon = 1$, the familiar worst-case bounds are recovered. For an outcome bounded in $[0, 1]$, these bounds have width equal to 1, which means that the sign of the treatment effect is not identified. Varying the sensitivity parameter to recover different identification regions has been proposed in other works, such as \cite{richardson2014nonparametric}, \cite{kennedy2019survivor} and \cite{diaz2013sensitivity}, albeit for different targets of inference or sensitivity models. 

An equivalent formulation of our model \eqref{eq:counterfactual_contamination} is one where there is a latent selection indicator $S \in \{0, 1\}$, with $\Pb(S = 1) = 1 - \epsilon$, such that $A \notind Y^a \mid \Xb, S = 0$, but $A \ind Y^a \mid \Xb, S = 1$. The following lemma rewrites $\psi$ in terms of $S$.
\medskip
\begin{lemma}\label{general_def_ate}
	Let $\lambda_a(\Xb) = \E(Y^{a}\mid A = 1 - a, \Xb, S = 0)$. Under consistency \eqref{assumption:consistency} and positivity \eqref{assumption:positivity}, it holds that
	\begin{align*}
	\psi = \E((1-S)[\{Y - \lambda_{1-A}(\Xb)\}(2A - 1)] + S\{\E( Y \mid  A = 1, \Xb, S = 1) - \E(Y \mid  A = 0, \Xb, S = 1)\})
	\end{align*}
\end{lemma}
All proofs can be found in the supplementary material. As shown in Lemma \ref{general_def_ate}, $\psi$ depends on three unobservable quantities: $\lambda_0(\Xb)$, $\lambda_1(\Xb)$ and $S$. The quantity $\lambda_1(\Xb)$ ($\lambda_0(\Xb)$) represents the average outcome for those control (treated) units subject to unmeasured confounding had they taken the treatment (control) instead. Without further assumptions, the observed distribution $\Pb$ would not impose any restrictions on $\lambda_0(\Xb)$ or $\lambda_1(\Xb)$ even if $S$ was known.

For any given $\epsilon$, a sharp lower (upper) bound on $\psi$ can be obtained by minimizing (maximizing) $\psi$ in Lemma \ref{general_def_ate} over $\lambda_0(\Xb)$, $\lambda_1(\Xb)$ and $S$. Without imposing some restrictions on the distribution of $S$, the optimization step involves finding, and nonparametrically estimating, the optimal regression functions $\E(Y \mid A = a, \Xb, S = 1)$. Given a sample of $n$ observations, this step would involve fitting regression functions on ${n \choose \lceil n \epsilon \rceil}$ different sub-samples of size $\lceil n \epsilon \rceil$, which is computationally very costly even for moderate sample sizes. 

Instead, we proceed by requiring that $S \ind (Y, A) \mid \Xb$; we call the resulting sensitivity model ``$X$-mixture model". The assumption that $S \ind (Y, A) \mid \Xb$ can be interpreted in at least three ways. First, one may hope that it holds exactly for the mechanism that generated the sample. For instance, it is trivially satisfied, for example, if $S$ is just a possibly unknown, deterministic function of the observed covariates. An example satisfying this condition is given by the selected ignorability framework proposed in \cite{joffe2010selective}: if the treatment is as good as randomized conditional on hematocrit (and possibly other observed covariates), then $S$ could be an indicator of whether hematocrit is missing. 

Even if it does not hold exactly, assuming $S \ind (Y, A) \mid \Xb$ may be a close approximation to reality that one can use to make the problem computationally tractable. This second interpretation is in the same spirit as using parametric regression models in order to simplify a given problem, hoping that they will be a close approximation to the true regression function. Third, even if $S \not\ind (Y, A) \mid \Xb$, the $X$-mixture model can help determining whether a study is not robust to unmeasured confounding. Because the bounds if no assumptions are made will be at least as wide as those under $S \ind (Y, A) \mid \Xb$, if a study does not appear robust in the $X$-mixture model, it will not appear robust in the general case either. In the following theorem, we derive closed-form expressions for sharp bounds on $\psi$ in the $X$-mixture model. 
\medskip
\begin{theorem}[Bounds in $X$-mixture model] \label{theorem_xmodel}
	Suppose that assumptions \ref{assumption:consistency} and \ref{assumption:positivity} hold. Further suppose that
	\begin{align*}
	S \ind (A, Y) \mid  \Xb \tag{A1} \label{x_mixture_assumption}
	\end{align*}
	and that $\Pb( Y \in [y_{\text{min}}, y_{\text{max}}] ) = 1$, for $y_{\text{min}}, y_{\text{max}}$ finite. Choose $\delta \in [0, 1]$ such that
	\begin{align}\label{eq_bias_delta}
	L_a \equiv \delta\{y_{\text{min}} - \mu_a(\Xb) \} \leq \lambda_a(\Xb) - \mu_a(\Xb) \leq \delta\{y_{\text{max}} - \mu_a(\Xb) \} \equiv U_a \text{ with prob. 1}
	\end{align}
	for $a \in \{0, 1\}$. Then, as a function of $\epsilon$, sharp bounds on $\psi$ are:
	\begin{align*}
	\psi_l(\epsilon) & = \E\left[\mu_1(\Xb) - \mu_0(\Xb) + \one\left\{ g(\etab) \leq q_{\epsilon} \right\} g(\etab) \right] - \epsilon \delta (y_{\text{max}} - y_{\text{min}})\\
	\psi_u(\epsilon) & = \E\left[\mu_1(\Xb) - \mu_0(\Xb) + \one\left\{ g(\etab) > q_{1 - \epsilon }\right\} g(\etab) \right]
	\end{align*}
	where $g(\etab) = \pi(0\mid \Xb)U_1 - \pi(1\mid \Xb)L_0$ and $q_{\tau}$ is its $\tau$-quantile.
\end{theorem}
Theorem \ref{theorem_xmodel} yields the identification of sharp lower and upper bounds on $\psi$ when it is suspected that $100\epsilon \%$ of the units in the sample are confounded and it is assumed that predicting whether a unit is confounded or not cannot be improved by conditioning on $(Y, A)$. Relaxing condition \eqref{x_mixture_assumption} to $S \ind Y \mid (A, \Xb)$ poses no additional challenges and it is discussed in Appendix \ref{appendix:xa-model}. We refer to this relaxed version of the $X$-mixture model as the ``$XA$-mixture model." Notably, it covers the differential effects framework of \cite{rosenbaum2006differential}, as one could specify $S = \one(A_1 + A_2 = 1)$ for some binary treatment $A_1$ and $A_2$. 

The bounds are in terms of the parameters $\epsilon$ and $\delta$, as well as the regression functions $\pi(a \mid \Xb)$ and $\mu_a(\Xb)$, and they involve non-smooth transformations of unknown functions of $\Pb$.	The parameter $\epsilon$ is our main sensitivity parameter and controls the proportion of unmeasured confounding in the sample. Parallely, $\delta$ controls the extent of unmeasured confounding among the $S = 0$ units, as it bounds the difference between the unobservable regression $\lambda_a(\Xb)$ and the estimable regression $\mu_a(\Xb)$. Notice that \eqref{eq_bias_delta} always holds for $\delta =  1$. Setting $\delta < 1$ imposes an untestable assumption on the severity of the unmeasured confounding, which might be sensible if some knowledge on the confounding mechanism is available. Specifically, our parametrization is such that $\lambda_a(\Xb)$ can be bounded by linear combinations of $y_{\text{min}}$, $y_{\text{max}}$ and $\mu_a(\Xb)$:
\begin{align*}
\delta y_{\text{min}} + (1-\delta) \mu_a(\Xb) \leq \lambda_a(\Xb) \leq \delta y_{\text{max}} + (1-\delta) \mu_a(\Xb)
\end{align*}
Unless otherwise specified, in what follows we consider $y_{\text{min}} = 0$, $ y_{\text{max}} = 1$ and set $\delta = 1$, thus yielding
\begin{align*}
\psi_l(\epsilon) & = \E\left[\mu_1(\Xb) - \mu_0(\Xb) + \one\left\{ g(\etab) \leq q_{\epsilon} \right\} g(\etab) \right] - \epsilon \\
\psi_u(\epsilon) & = \E\left[\mu_1(\Xb) - \mu_0(\Xb) + \one\left\{ g(\etab) > q_{ 1 - \epsilon }\right\} g(\etab) \right]
\end{align*}
for $g(\etab) = \pi(0\mid \Xb)\{1-\mu_1(\Xb)\} + \pi(1\mid \Xb) \mu_0(\Xb)$.
If $Y$ is bounded, this choice does not impose any assumption since $Y$ can be rescaled to be in $[0, 1]$. If $Y$ is unbounded, Theorem \ref{theorem_xmodel} is not directly applicable, but a similar result can be derived if one is willing to assume that $| \lambda_a(\Xb) - \mu_a(\Xb)| \ \leq \delta$ for $a \in \{0, 1\}$ and $\delta < \infty$. We leave further investigation of the unbounded case as future work. We conclude this section with four remarks aiming to shed some more light on the bounds derived in Theorem \ref{theorem_xmodel}. 
\medskip
\begin{remark}
	Suppose $Y$ is bounded in $[0, 1]$ and take $\delta =  1$. The length of the bound is then
	\begin{align*}
	\Delta(\epsilon) = [\E\{g(\etab) \mid g(\etab) > q_{ 1 - \epsilon }\} - \E \{ g(\etab) \mid g(\etab) \leq q_{ \epsilon }\} + 1]\epsilon
	\end{align*}
	If $S$ was known, the length of the bound would reduce to $\Delta(\epsilon) = \epsilon$. Thus, we can view the term $[\E\{ g(\etab) \mid g(\etab) > q_{ 1 - \epsilon }\} - \E\{g(\etab) \mid g(\etab) \leq q_{ \epsilon }\}]\epsilon$ as the ``cost" of not knowing who the confounded units are. 
\end{remark}
\medskip
\begin{remark}\label{remark_robust_stat}
	The conditional independence of $S$ and $Y$ considerably simplifies the optimization step. To see this, notice that $\E(Y \mid A = a, \Xb, S = 1) = \mu_a(\Xb)$ if $S \ind Y \mid A, \Xb$. In turn, this implies that $\psi$ can be written as
	\begin{align*}
	\psi = \E[\Gamma(Y, A, \Xb) + S \{\mu_1(\Xb) - \mu_0(\Xb) - \Gamma(Y, A, \Xb)\}]
	\end{align*}
	where $\Gamma(Y, A, \Xb) = \{Y - \lambda_{1-A}(\Xb)\}(2A - 1)$. Therefore, bounds on $\psi$ can be derived from bounds on $~\E\left\{\mu_1(\Xb) - \mu_0(\Xb) - \Gamma(Y, A, \Xb) \mid S = 1\right\}$, which fits the framework studied by \cite{horowitz1995identification}. In their work, the goal is to do inference about a distribution $Q_1$ using data $Y$ such that $Y = ZY_1 + (1-Z)Y_0$, with $Z \in \{0, 1\}$ and $Y_i \sim Q_i$. They discuss two models: the ``contaminated sampling model", which assumes $Z$ to be independent of $Y_1$, and the ``corrupted sampling model", which does not make this assumption. If it is known that $\Pb(Z = 0) \leq \lambda$, they derive sharp bounds on the conditional expectation of $Y_1$ given some covariates $\Xb$ when contamination or corruption does not occur in $\Xb$. Our setup does not immediately fit this framework because corruption applies to all observed variables $(Y, A, \Xb)$. However, if $S \ind Y \mid A, \Xb$, the optimal solution for $S$ can be found by considering only the marginal distribution of the one-dimensional random variable $\mu_1(\Xb) - \mu_0(\Xb) - \Gamma(Y, A, \Xb)$. Following the terminology in \cite{horowitz1995identification}, we may view the assumption that $S \ind Y \mid A, \Xb$ as a compromise between contamination ($S \ind (Y, A , \Xb)$) and corruption (no assumption on $S$).
\end{remark}
\medskip
\begin{remark}
	As pointed out by \cite{robins2002covariance}, many interesting sensitivity analyses make use of parameters that depend on the covariates collected. In turn, this might hinder the direct comparison of studies' robustness. For example, a study where many confounders have been properly taken into account might appear more sensitive to departures from the no-unmeasured-confounding assumption than a study that failed to control for any confounder. This could happen, for instance, if the effect estimate in the former study is closer to the null value than the estimate from the latter. This apparent paradox might arise because a sensitivity analysis measures departures from a weak or strong assumption depending on whether many or few observed confounders are collected. Our proposed sensitivity analysis hinges on $\epsilon$, the proportion of unmeasured confounding, which depends on the covariates collected. As such, it might be subject to this paradox.
\end{remark}
\medskip
\begin{remark}
	Section 4  of \cite{rosenbaum1987sensitivity} contains a modification to the sensitivity analysis proposed in that paper, and briefly summarized in our introduction, that allows an unknown fraction $\beta$ of the sample to suffer from arbitrarily confounding. While conceptually similar to the approach presented in this paper, their method relies on exact matching. In fact, if units are exactly matched on observed covariates, our sensitivity model recovers Rosenbaum's with $\beta = \epsilon$ and $\Gamma = 0$. However, exact matching is often infeasible due to the presence of continuous or high-dimensional covariates. Therefore, our work can be viewed as an extension to Rosenbaum's Section 4 model to the case where units are not matched on observed covariates. 
\end{remark}
\subsection{One-number Summary of a Study's Robustness}\label{section:eps0}
In practice, one might want to report a one-number summary of how robust the estimated effect is to the number of confounded units. An example of such summary is the minimum proportion of confounded units $\epsilon_0$ such that the bounds on $\psi$ are no longer informative about the sign of the effect, i.e. that they contain zero. Larger values of $\epsilon_0$ indicate that the estimated effect is more robust to potential unmeasured confounding. Mathematically, 
\begin{align*}
\epsilon_0 = \argmin_{\epsilon \in \Ec} \one[\sgn\{\psi_l(\epsilon)\} \neq \sgn\{\psi_u(\epsilon)\}]
\end{align*}
where $\sgn(x)$ measures the sign of $x$, $\sgn(x) = -\one (x < 0) + \one(x > 0)$. Because $\psi_u(\epsilon = 1) - \psi_l(\epsilon = 1) = 1$, the minimum is guaranteed to be attained in $\Ec = [0, 1]$. Furthermore, under certain mild conditions, the bounds are continuous and strictly monotone in $\epsilon$, hence $\epsilon_0$ is generally the unique value such that $\psi_l(\epsilon_0) = 0$ or $\psi_u(\epsilon_0) = 0$. This motivates the moment condition $\psi_l(\epsilon_0) \psi_u(\epsilon_0) = 0$, which we use to construct a $Z$-estimator of $\epsilon_0$.

Other authors have proposed one-number summaries of a study's robustness to unmeasured confounding. For example, the minimum value for $\Gamma$ in Rosenbaum's framework and its extensions \citep{rosenbaum1987sensitivity, gastwirth1998dual, zhao2017sensitivity, yadlowsky2018bounds} such that the observed effect ceases to be statistically significant can be used as a summary of study's robustness to unmeasured confounding. Recently, \cite{ding2016sensitivity} and \cite{vanderweele2017sensitivity} have introduced the E-Value, which measures the minimum strength of association, on the risk ratio scale, that an unmeasured confounder would need to have with both the outcome and the treatment in order to ``explain away" the observed effect of the treatment on the outcome. In order to derive the elegant formula for the E-Value, the unobserved confounder is assumed to be associated with the treatment and with the outcome in equal magnitude. Furthermore, the derivation makes use of a bounding factor that needs to be computed for each stratum of the covariates. Computing such bounding factor when the observed covariates are continuous or high-dimensional can be problematic. Moreover, their method requires additional approximations if the outcome is not binary. On the other hand, the one-number summary proposed here does not require any further assumption other than the restriction on $S$ described above. Hence, we view these summary measures as complementary and the specific context would generally dictate which one is more appropriate.
\section{Estimation \& Inference}
\subsection{Proposed Estimators}\label{section:estimators}
There are at least two types of bias that can arise when estimating a causal effect using observational data: the bias arising from incorrectly assuming that all confounders have been collected and the statistical bias of the chosen estimator \citep{luedtke2015statistics}. In Section \ref{section:model}, we constructed a model to probe the effects of the former bias. In this section, we propose estimators that aim to minimize the latter. Our estimators of the bounds are built using the efficient influence functions (IFs) and cross-fitting. IFs play a crucial role in nonparametric efficiency theory, as the variance of the efficient IF can be considered the nonparametric counterpart of the Cramer-Rao lower bound in parametric models. Furthermore, estimators constructed using the efficient IF have favorable properties, such as doubly-robustness or second-order bias. Here, we note that $\psi_l(\epsilon)$ and $\psi_u(\epsilon)$ do not possess an influence function, as they are not pathwise differentiable. However, certain terms appearing in their expressions, such as $\E\{\mu_a(\Xb)\}$, are pathwise differentiable; as such, they can be estimated using IFs. For terms that are not pathwise differentiable we resort to plug-in estimators. We refer to \cite{bickel1993efficient}, \cite{vdvsemiparametric2002}, \cite{van2003unified}, \cite{tsiatis2007semiparametric}, \cite{chernozhukov2016double} and others for detailed accounts on IFs and their use.

To ease the notation in this section, let
\begin{align*}
\nu(\Ob; \etab) = \frac{(2A - 1)\left\{Y - \mu_A(\Xb) \right\}}{\pi(A\mid \Xb)} + \mu_1(\Xb) - \mu_0(\Xb)
\end{align*} 
denote the uncentered influence function for the parameter $\E\left\{ \mu_1(\Xb) - \mu_0(\Xb) \right\}$. Furthermore, let $\tau(\mathbf{O}; \etab)$ denote the uncentered influence function for $\E\left\{g(\etab)\right\}$:
\begin{align*}
\tau(\mathbf{O}; \etab) = \frac{(1-2A)\left\{Y - \mu_A(\Xb)\right\}}{\pi(A\mid \Xb) / \pi(1-A\mid \Xb)}+ A \mu_0(\Xb) + (1-A)\left(1 - \mu_1(\Xb)\right)
\end{align*}
and let
\begin{align*}
& \varphi_l(\Ob; \etab; q_\epsilon) = \nu(\Ob; \etab) + \one\{ g(\etab) \leq q_\epsilon \}\tau(\Ob; \etab) - \epsilon \\
& \varphi_u(\Ob; \etab; q_{1-\epsilon}) = \nu(\Ob; \etab) + \one\{ g(\etab) > q_{1-\epsilon} \}\tau(\Ob; \etab)
\end{align*}
Then, it holds that $\psi_l(\epsilon) = \E \{ \varphi_l(\Ob; \etab; q_\epsilon) \}$ and $\psi_u(\epsilon) = \E\{ \varphi_u(\Ob; \etab; q_{1-\epsilon}) \}$.

Following \cite{robins2008higher}, \cite{zheng2010asymptotic} and \cite{chernozhukov2016double} among others, we use cross-fitting to allow for arbitrarily complex estimators of the nuisance functions $\etab$ and $q_\tau$ in order to avoid empirical process conditions. Specifically, we split the data into $B$ disjoint groups of size $n / B$ and we let $K_i = k$ indicate that subject $i$ is split into group $k$, for $k \in \{1, \ldots, B\}$. Notice that it is not required that the groups have equal size, for example each $K_i$ could be drawn uniformly from $\{1, \ldots, B\}$. For simplicity, we proceed with having equal-size groups. We let $\Pn$ denote the empirical measure as $\Pn\left\{f(\mathbf{O})\right\} = \frac{1}{n}\sum_{i = 1}^n f(\mathbf{O}_i)$ and $\Pn^k$ denote the sub-empirical measure as $\Pn^k\left\{f(\mathbf{O})\right\} = \sum_{i = 1}^n f(\mathbf{O}_i) \one(K_i = k) / \sum_{i = 1}^n \one(K_i = k)$. In addition, we let $\widehat{\etab}_{-k}$ denote the estimator of $\etab$ computed without using observations from fold $K = k$ and $\widehat{q}_{\tau, -k}$ denote the estimator of $q_\tau$ equal to the empirical quantile of $g(\widehat{\etab}_{-k})$ solving $\Pn^k[\one\{ g(\widehat{\etab}_{-k}) \leq \widehat{q}_{\tau, -k} \}] = \tau + o_\Pb(n^{-1/2})$. Then, we estimate the bounds as
\begin{align*}
\widehat{\psi}_l(\epsilon) & = \frac{1}{B}\sum_{k = 1}^B \Pn^k [\nu(\Ob; \widehat{\etab}_{-k}) + \one\{g(\widehat{\etab}_{-k}) \leq \widehat{q}_{\epsilon, -k}\}\tau(\mathbf{O}; \widehat{\etab}_{-k})] - \epsilon \nonumber  \equiv \Pn\left\{ \varphi_l(\mathbf{O}; \widehat{\etab}_{-K}, \widehat{q}_{-K, \epsilon}) \right\} \\
\widehat{\psi}_u(\epsilon) & = \frac{1}{B}\sum_{k = 1}^B \Pn^k [\nu(\Ob; \widehat{\etab}_{-k}) + \one\{g(\widehat{\etab}_{-k}) > \widehat{q}_{1-\epsilon, -k}\} \tau(\mathbf{O}; \widehat{\etab}_{-k})]\equiv \Pn\left\{ \varphi_u(\mathbf{O}; \widehat{\etab}_{-K}, \widehat{q}_{-K, 1-\epsilon}) \right\} 
\end{align*}
The computation of the estimators above is straightforward as it amounts to fitting regression functions on $B-1$ subsets of the data and evaluate the estimated functions at the values of the covariates on the corresponding test set. The use of cross-fitting lends itself naturally to the use of parallel computing as one can estimate the regression functions on different subsets of the data simultaneously. We incorporate this possibility in our implementation of the methods in \texttt{R}. Moreover, it is worth noting that cross-fitting does not discard any data point in the estimation step, since each observation is used twice without overfitting: once for estimating the regression functions and once for estimating the expectation operator. In addition, because we are working under a fully nonparametric model, there exists only one influence function; therefore, our estimators of the pathwise differentiable terms are efficient in the sense that they  asymptotically achieve the semiparametric efficiency bound. 

Finally, while the estimators of the bounds discussed in this section have several attractive properties in terms of computational tractability and convergence rates, they might not be monotone in $\epsilon$ in finite samples. To remedy this, the estimators can be ``rearranged" using the procedure described in \cite{chernozhukov2009improving}. We apply this procedure in Section \ref{section:illustrations}, although we find that the original, non-rearranged estimators achieve low bias and nominal uniform coverage as well. 
\subsection{Establishing Weak Convergence} \label{weak_convergece}
To state asymptotic guarantees for the proposed estimators, we first make the following technical assumption:
\medskip
\begin{assumption}[Margin Condition] \label{margin_condition}
	The random variable $g(\etab)$ has absolutely continuous CDF and there exists $\alpha > 0$ such that for all $t > 0$ and $\tau \in \Ec$, it holds that $\Pb \left( \left|  g\left(\etab\right) - q_\tau \right| \leq t  \right) \lesssim t^\alpha$ and $\Pb \left( \left|  g\left(\etab\right) - q_{1-\tau} \right| \leq t  \right) \lesssim t^\alpha$. 
\end{assumption}
Assumption \ref{margin_condition} requires that there is not too much mass around any $\epsilon$-quantile or ($1-\epsilon$)-quantile of $g(\etab)$, for $\epsilon \in \Ec$. It is essentially equivalent to the margin condition used in classification problems \citep{audibert2007fast}, optimal treatment regime settings \citep{luedtke2016statistical, van2014targeted}, and other problems involving estimation of non-smooth functionals \citep{kennedy2018sharp, kennedy2019survivor}. Notably it is satisfied for $\alpha = 1$ if, for instance, the density of $g(\etab)$ is bounded on $\Ec$. We give the main convergence theorem for $\widehat{\psi}_u(\epsilon)$. A similar statement holds for $\widehat{\psi}_l(\epsilon)$.
\medskip
\begin{theorem}\label{gp_convergence}
	Let 
	\begin{align*}
	\widehat{\sigma}^2_u(\epsilon) = \Pn \{( \varphi_u(\Ob; \widehat{\etab}_{-K}, \widehat{q}_{1-\epsilon, -K}) - \widehat{\psi}_u(\epsilon) - \widehat{q}_{1-\epsilon, -K}[ \one\{g(\widehat{\etab}_{-K}) > \widehat{q}_{1-\epsilon, -K} \}-  \epsilon ])^2\}
	\end{align*}
	be the estimator of the variance function 
	\begin{align*}
	\sigma^2_u(\epsilon) = \E \{( \varphi_u(\Ob; \etab, q_{1-\epsilon}) - \psi_u(\epsilon) -q_{1-\epsilon}[\one\{g(\etab) > q_{1-\epsilon} \} - \epsilon] )^2\}
	\end{align*}
	If assumptions \ref{assumption:consistency}, \ref{assumption:positivity} and \ref{margin_condition} hold, and the following conditions also hold:
	\begin{enumerate}
		\item \label{gp_convergence_condpos} $\Pb \left\{t \leq \widehat{\pi}(a\mid \mathbf{X}) \leq 1 - t \right\} = 1$ for $a = 0, 1$ and some $t > 0$. 
		\item \label{gp_convergence_condvar} $\sup_{\epsilon \in \Ec} \left| \frac{\widehat{\sigma}_u(\epsilon)}{\sigma_u(\epsilon)} - 1 \right| = o_\Pb(1)$. 
		\item  \label{gp_convergence_cond_envelope} $\| \sup_{\epsilon \in \Ec} \ | \varphi_u(\mathbf{o}; \widehat{\etab}, \widehat{q}_{1-\epsilon}) - \varphi_u(\mathbf{o}; \etab, q_{1-\epsilon}) - q_{1-\epsilon}[\one\{g(\widehat{\etab}) > \widehat{q}_{1-\epsilon} \} - \one\{g(\etab) > q_{1-\epsilon} \}] | \| \ = o_\Pb(1)$.
		\item \label{gp_convergence_condLinf} $\left( \left\|  g(\widehat{\etab}) - g\left(\etab\right) \right\|_\infty + \sup_{\epsilon \in \Ec} \left|  \widehat{q}_{1-\epsilon} - q_{1-\epsilon} \right|  \right)^{1+\alpha} = o_\Pb(n^{-1/2})$, for $\alpha$ satisfying assumption \ref{margin_condition}. 
		\item \label{gp_convergence_condL2} $\left\| \widehat{\pi}(1\mid \mathbf{X}) - \pi(1\mid \mathbf{X}) \right\| \max_a \|  \widehat{\mu}_a(\mathbf{X}) - \mu_a(\mathbf{X}) \| \ = o_\Pb(n^{-1/2})$.
	\end{enumerate}
	Then $\sqrt{n}\{\widehat{\psi}_u(\epsilon) - \psi_u(\epsilon)\}/\widehat{\sigma}_u(\epsilon) \indist \Gb(\epsilon)$ in $\ell^{\infty}(\Ec)$, with $\Ec \subseteq [0, 1]$, where $\Gb(\cdot)$ is a mean-zero Gaussian process with covariance $\E\left\{\Gb(\epsilon_1)\Gb(\epsilon_2)\right\} = \E\left\{\phi_u(\Ob; \etab, q_{1-\epsilon_1})\phi_u(\Ob; \etab, q_{1-\epsilon_2})\right\}$ and 
	\begin{align*}
	\phi_u(\mathbf{O}; \etab, q_{1-\epsilon}) = \frac{\varphi_u(\Ob; \etab, q_{1-\epsilon}) - \psi_u(\epsilon) -  q_{1-\epsilon}[\one\{g(\etab) > q_{1-\epsilon}\} - \epsilon]}{\sigma_u(\epsilon)}.
	\end{align*}
\end{theorem}
Theorem \ref{gp_convergence} gives sufficient conditions so that the estimated curves tracing the lower and upper bounds as a function of $\epsilon$ converge to a Gaussian process. In turn, this enables the computation of confidence bands trapping the average treatment effect with any desired confidence level uniformly over $\epsilon$. The first three conditions of the theorem are quite mild. Condition \ref{gp_convergence_condpos} is a positivity condition requiring that the estimator of the propensity score is bounded away from 0 and 1. Condition \ref{gp_convergence_condvar} requires uniform consistency of the variance estimator at any rate. Condition \ref{gp_convergence_cond_envelope} holds if, in addition to satisfying the margin assumption \ref{margin_condition}, $g(\widehat{\etab})$ and $\widehat{q}_\tau$ converge uniformly, in $\xb$ and $\epsilon$ respectively, to the truth at any rate. 

The key assumptions are conditions \ref{gp_convergence_condLinf} and \ref{gp_convergence_condL2}. While more restrictive than the first three, these conditions can be satisfied even if flexible machine learning tools are used. In fact, condition \ref{gp_convergence_condL2} only requires that the product of the $L_2$ errors in estimating $\pi(a \mid \Xb)$ and $\mu_a(\Xb)$ is of order $n^{-1/2}$, which means that, for example, each regression function can be estimated at the slower rate $n^{-1/4}$. A rate of convergence in $L_\infty$ norm of order $n^{-1/4}$ is also sufficient to satisfy condition \ref{gp_convergence_condLinf} if the density of $g(\etab)$ is bounded because the margin assumption \ref{margin_condition} would hold for $\alpha = 1$. A convergence rate of order $n^{-1/4}$ can be achieved if nonparametric smoothness, sparsity or other structural assumptions are imposed on the true regression functions. For instance, if a minimax optimal estimator is used, in order to satisfy condition \ref{gp_convergence_condL2}, it is sufficient that the underlying regression functions belong to a $\beta$-H\"older class with smoothness parameter $\beta > p/2$, where $p$ is the number of covariates. In addition, even in regimes of very large $p$, convergence at $n^{-1/4}$ rate can be achieved under structural assumptions such as additivity or sparsity \citep{horowitz2009semiparametric, raskutti2012minimax, farrell2015robust, yang2015minimax, kandasamy2016additive}. Furthermore, such convergence rate can also be achieved if the regression functions belong to the class of cadlag functions with bounded variation norm \citep{benkeser2016highly, van2017generally}. We refer to \cite{gyorfi2006distribution} among others for additional convergence results. 

Similarly to \cite{kennedy2018nonparametric}, we can use Theorem \ref{gp_convergence} and the multiplier bootstrap to construct uniform confidence bands covering the identification region $[\psi_l(\epsilon), \psi_u(\epsilon)]$. Placing $(1-\alpha/2)$ uniform confidence bands on each curve also yields a (conservative) $(1-\alpha)$ uniform confidence band for $\psi$. We also deploy the procedure of \cite{imbensmanskyci} to construct bands covering just $\psi$ that are valid pointwise. Details are provided in Appendix \ref{section:uniform_bands}. Constructing uniformly valid bands covering $\psi$, as opposed to the whole identification region, is left for future research. 
\subsection{Estimation of the One-Number Summary $\epsilon_0$}\label{section:est_eps0}	
In our settings, a natural way to define $\epsilon_0$ is via the moment condition $\psi_l(\epsilon_0) \psi_u(\epsilon_0) = 0$ and construct an estimator $\widehat{\epsilon}_0$ defined implicitly as the solution to the empirical moment condition 
\begin{align*}
\Pn\{ \varphi_l(\mathbf{O}; \widehat{\etab}_{-K}, \widehat{q}_{\widehat{\epsilon}_0, -K})\} \Pn\{ \varphi_u(\mathbf{O}; \widehat{\etab}_{-K}, \widehat{q}_{1-\widehat{\epsilon}_0, -K}) \} = o_\Pb(n^{-1/2}).
\end{align*}
Standard results in $Z$-estimation theory (Theorem 3.3.1 in \cite{van1996weak}) yield the following theorem.
\medskip
\begin{theorem}\label{eps_tilde_theorem}
	Suppose that the CDF $G$ of $g(\etab)$ is strictly increasing in neighborhoods of $q_{\epsilon_0}$ and $q_{1-\epsilon_0}$. Suppose assumptions \ref{assumption:consistency}, \ref{assumption:positivity}, \ref{margin_condition} and conditions 1, 3, 4, 5 (and 3's and 4's counterpart for the lower bound) from Theorem \ref{gp_convergence} are satisfied with $\Ec = [0, 1]$.	Then
	\begin{align*}
	\sqrt{n} \left( \widehat{\epsilon}_0 - \epsilon_0 \right) \indist N\left(0, \left[\psi_u(\epsilon_0)(q_{\epsilon_0} - 1) + \psi_l(\epsilon_0) q_{1-\epsilon_0}\right]^{-2} \text{var}\left\{ \tilde{\varphi} (\epsilon_0) \right\}\right)
	\end{align*}
	provided that the denominator $\psi_u(\epsilon_0)(q_{\epsilon_0} - 1) + \psi_l(\epsilon_0)q_{1-\epsilon_0} \neq 0$, and where the unscaled influence function is
	\begin{align*}
	\tilde{\varphi}(\epsilon_0) = \psi_u(\epsilon_0)[\varphi_l(\Ob; \etab, q_{\epsilon_0}) - q_{\epsilon_0} \one\{g(\etab) \leq q_{\epsilon_0}\}] + \psi_l(\epsilon_0)[\varphi_u(\Ob; \etab, q_{1-\epsilon_0}) - q_{1-\epsilon_0}  \one\{g(\etab) > q_{1-\epsilon_0}\}].
	\end{align*}
\end{theorem}
Theorem \ref{eps_tilde_theorem} describes sufficient conditions so that $\widehat{\epsilon}_0$ is $\sqrt{n}$-consistent and asymptotically normally distributed. We require the same conditions as the ones required for Theorem \ref{gp_convergence}, plus that the CDF of $g(\etab)$ is strictly increasing in neighborhoods of $q_{\epsilon_0}$ and $q_{1-\epsilon_0}$. The asymptotic normality of $\widehat{\epsilon}_0$ relies on the existence (and non-singularity) of the derivative of the map $\epsilon \mapsto \psi_l(\epsilon)\psi_u(\epsilon)$ at $\epsilon = \epsilon_0$. Calculating such derivative requires computing the derivative of the quantile function, which is why we require the CDF of $g(\etab)$ to be strictly increasing in the relevant neighborhoods. We expect all these conditions to be satisfied in practice in the presence of continuous covariates and enough smoothness or sparsity for the regression functions.\footnote{In principle, one could construct the empirical moment condition after performing the rearrangement procedure of \cite{chernozhukov2009improving}. Whether or not the rearrangement is done, we expect the inference about $\epsilon_0$ to be equivalent asymptotically and vary minimally in finite samples.} Asymptotic normality allows the straightforward calculation of a Wald-type confidence interval for $\epsilon_0$ using a consistent estimate for the variance. We thus propose reporting both a point-estimate for $\epsilon_0$ and $1-\alpha$ confidence interval as a summary of the study's robustness to unmeasured confounding.\footnote{In order to incorporate finite sampling uncertainty in sensitivity analyses, one-number summaries of a study's robustness are generally computed as the values of the sensitivity parameter(s) such that a $\alpha$-level confidence interval for the ATE under no unmeasured confounding includes the null value. Choosing different $\alpha$s to estimate the ATE with no residual confounding may then yield different conclusions regarding the study's robustness to unmeasured confounding, despite the latter being a separate inferential task. Constructing a confidence interval for $\epsilon_0$ directly bypasses this issue.}
\section{Illustrations}\label{section:illustrations}
\subsection{Simulation Study} \label{section:sim}
In this section, we report the results of the simulations we performed to investigate the finite-sample performance of our proposed estimators. We consider the following data generating mechanism:
\begin{align*}
& X_i \sim \text{TruncNorm}(\mu=0, \sigma=1, \text{lb}=-2, \text{ub}=2) \text{ for } i \in \{1, 2\}, \ U \sim \text{Bern}(0.5), \\
& S \mid X_1, X_2 \sim \text{Bern}\{\Phi(X_1)\},
\\ & A \mid X_1, X_2, U, S \sim \text{Bern}[0.5\{\Phi(X_1) + 0.5S + (1-S)U\}], \\
& Y^{a} \mid X_1, X_2, U, S, A \sim \text{Bern}\{0.25 + 0.5\Phi(X_1 + X_2) + (a - 0.5)r - 0.1 U\}, \\
& Y = AY^1 + (1-A)Y^0,
\end{align*}
where $\Phi(\cdot)$ denotes the CDF of a standard normal random variable. Notice that 
\begin{align*}
\Pb(A = 1 \mid X_1, X_2, S = 0) = \Pb(A = 1 \mid X_1, X_2, S = 1) = 0.5\Phi(X_1) + 0.25,
\end{align*}
thus this model satisfies the assumptions of Theorem \ref{gp_convergence} and it implies that $\E(Y^{1} - Y^{0})~=~r$. The random variable $U$ acts as a binary unmeasured confounder; given the observed covariates $\Xb$, units with $S = 0$ and $U = 1$ are more likely to be treated and exhibit $Y = 0$ than those with $S = 0$ and $U = 0$. Therefore, under this setup, one would expect the treatment effect to be underestimated if the no-unmeasured-confounding assumption is (incorrectly) assumed to be true.\footnote{In the context of the toy example of Section \ref{sec:motivation}, $U$ and $S$ indicate whether the parents are smokers and whether they would smoke at home respectively, $X_1$ and $X_2$ may be measures of the parents' education level and income respectively, $A$ indicates adolescent alcohol consumption and $Y$ indicates the occurrence of liver disease.}

We estimate the lower bound $\psi_l(\epsilon)$, the upper bound $\psi_u(\epsilon)$ and $\epsilon_0$ using the methods outlined in Section \ref{section:estimators}. In particular, we use 5-fold cross-fitting to estimate the nuisance functions, fitting both generalized linear and additive models via the SuperLearner method \citep{van2007super}. The performance of the proposed estimators is evaluated via integrated bias, root-mean-squared-error (RMSE), and coverage. These evaluation metrics offer insight into what sample size is required to achieve a good performance of the multiplier bootstrap, which relies on the convergence of the bounds' estimators to a Gaussian process.
\begin{align*}
\widehat{\text{bias}} = \frac{1}{I} \sum_{i=1}^I \left| \frac{1}{J} \sum_{j=1}^J \{\widehat{\psi}_{l, j}(\epsilon_i) - \psi_{l, j}(\epsilon_i) \} \right|, \quad \widehat{\text{RMSE}} = \frac{1}{I} \sum_{i=1}^I \left[ \frac{1}{J} \sum_{j=1}^J \{\widehat{\psi}_{l, j}(\epsilon_i) - T_j(\epsilon_i) \}^2 \right]^{1/2}
\end{align*}
and suitably modified formulas for $\psi_u(\epsilon)$ and $\epsilon_0$. We run $J = 500$ simulations across $I=21$ values of $\epsilon$ equally spaced in $\Ec = [0, 0.2]$. To better estimate $\epsilon_0$ we make the grid finer and consider $201$ values of $\epsilon$ equally spaced in $\Ec$. To evaluate 95\% uniform coverage, we say that the uniform band covers if it contains the true region $[\psi_l(\epsilon), \psi_u(\epsilon)]$ for all $\epsilon \in \Ec$. Finally, we assess bias and 95\% coverage for $\epsilon_0$. 
\begin{table}[H]
\centering
\begin{tabular}{c|ccc|ccc|cc}
	\multicolumn{1}{c}{$n$} & \multicolumn{3}{c}{Bias ($\times 100$) } & \multicolumn{3}{c}{$\sqrt{n} \times $RMSE} & \multicolumn{2}{c}{Coverage ($\times 100$)} \\
	\hline\hline
	& $\psi_l(\epsilon)$ & $\psi_u(\epsilon)$ & $\epsilon_0$ & $\psi_l(\epsilon)$ & $\psi_u(\epsilon)$ & $\epsilon_0$ & $[\psi_l(\epsilon), \psi_u(\epsilon)]$ & $\epsilon_0$ \\
	\hline
	500 & 0.38 & 0.12 & 2.47 & 0.95 & 0.96 & 1.32 & 95.4 & 97.0 \\
	1000 & 0.51 & 0.14 & 1.59 & 0.95 & 0.95 & 1.45 & 93.2 & 95.6 \\
	5000 & 0.04 & 0.10 & 0.16 & 0.99 & 0.98 & 1.72 & 92.4 & 95.4 \\
	10000 & 0.05 & 0.09 & 0.07 & 0.95 & 0.96 & 1.75 & 93.6 & 94.8 \\
	\hline\hline
\end{tabular}
\caption{\label{table1:simulation_res} Simulation results across 500 simulations.}
\end{table}
Table \ref{table1:simulation_res} shows the results of our simulation for $r = 0.05$. This set up is such that $\epsilon_0 = 0.041$. In addition, if no-unmeasured-confounding is erroneously thought to hold ($\epsilon = 0$), $\psi$ is, on average, underestimated since $\E\left\{ \mu_1(\Xb) - \mu_0(\Xb) \right\} \approx 0.023 < r$.
This simple simulation setup exemplifies what our theory predicts. Even for moderate sample sizes, we achieve approximately correct nominal uniform coverage for the identification region and $\epsilon_0$. Furthermore, the $\sqrt{n}\times$RMSE remains roughly constant as the sample size increases. Finally, in Section \ref{sec:power_sensitivity} of the Appendix, we extend this simulation study to investigate how conservative our model would be if the true $\epsilon_0$ is actually zero, i.e. there is no unmeasured confounding.

\subsection{Application} \label{section:data_analysis}
In this section, we illustrate the proposed sensitivity analysis by reanalyzing the data from the study on Right Heart Catheterization (RHC) conducted by \cite{connors1996effectiveness}.\footnote{Available at \url{http://biostat.mc.vanderbilt.edu/wiki/Main/DataSets}.} The data consist of 5735 records from critically ill adult patients receiving care in an ICU for certain disease categories in one out of five US teaching hospitals between 1989 and 1994. For each patient, demographic variables, comorbitidies and diagnosis variables as well as several laboratory values were recorded. A total of 2184 patients underwent RHC within the first 24 hours in the ICU. Within 30 days of admission, 1918 patients died, approximately $38.00\%$ and $30.64\%$ of the treated and control groups respectively. After conditioning on the measured confounders, the authors concluded that patients treated with RHC had, on average, lower probability of surviving (30-day mortality: $\text{OR}=1.24$, $95\% \text{ CI}=[1.03, 1.49]$). Notably, sensitivity analyses targeting potential violations of the propensity score model suggested robustness of the study's conclusions to unmeasured confounding.

We investigate the effects of varying the proportion of confounded units while avoiding any parametric assumptions on the nuisance regression functions. One reason to believe that a fraction of the sample might be effectively unconfounded is the following. Suppose there are two types of surgeons: those who prefer performing RHC (R-surgeon) and those who don't (NR-surgeon). One might believe that the surgeon's preference for RHC is a valid instrument. Roughly, an instrument is a variable that is unconfounded, associated with the treatment receipt, and that affects the outcome only through the treatment. It appears plausible that a surgeon's preference for RHC would satisfy these conditions if, for instance, the efficacy of RHC was not well understood at the time the study was conducted. In fact, physicians' preferences for a treatment have been used as IVs before, see for example \cite{hernan2006instruments} and \cite{baiocchi2014instrumental} for reviews and discussions. Then, the patients who would undergo RHC if assigned to an R-surgeon but would not undergo RHC if assigned to a NR-surgeon represent the unconfounded unknown fraction of the sample. 

Consider the group of patients who underwent RHC. A unit in this group can be either a ``complier" or a ``non-complier". She's a complier if she would not have undergone RHC if assigned to an NR-surgeon, whereas she's a non-complier if she would have undergone RHC regardless of the type of surgeon or only if assigned to a NR-surgeon. In many instances, these two types will differ in terms of observed covariates $\Xb$. However, for certain values $\xb$ of $\Xb$, a unit might be either a complier or a non-complier with non-zero probability. In this scenario, our relaxed $XA$-model posits that the probability of survival conditional on receiving RHC is the same for a complier and a non-complier sharing the same $\Xb = \xb$. Notice that this is not imposing any assumption on what would have happened to the non-complier had she not been treated. In fact, we derived the lower (upper) bound on the average effect of RHC by assuming that she would have certainly survived (died) had she not undergone RHC. This maximal conservativeness in deriving the bounds likely protects our conclusions from mild violations of our $X$- and $XA$-models.  

To construct the curves tracing the bounds using the data, we estimate the nuisance regression functions via the cross-validation-based SuperLearner ensemble \citep{van2007super}, combining generalized additive models, random forests, splines, support vector machines as well as generalized linear models. We perform 5-fold cross-fitting. We also construct pointwise and uniform confidence bands. Results are reported in Figure \ref{connors_results}.

In line with the results in \cite{connors1996effectiveness}, if no-unmeasured-confounding holds, patients treated with RHC show a statistically significant decrease in 30-day survival rates. The risk difference equals $-3.74\%$ ($95\% \text{ CI}=[-6.00\%, -1.49\%]$). Under the $X$-mixture model, the bounds on the difference in survival rate would include zero if more than $4.89\%$ ($95\% \text{ CI}=[1.50\%, 8.28\%]$) of the patients were confounded. The value reduces to $4.02\%$ ($95\% \text{ CI}=[1.59\%, 6.45\%]$) under the relaxed $XA$-mixture model. Whether robustness to $5\%$ of potentially confounded units is enough to attach a causal interpretation to the study's result largely depends on subject-matter knowledge. Earlier we have described $\epsilon_0$ as the proportion of ``non-compliers," but other interpretations are also possible. For instance, suppose it is known that, before deciding whether a patient undergoes RHC, most surgeons look at lab value $v_1$, but some may check lab value $v_2$ as well. Both values are correlated with survival, but only $v_1$ is measured. If reviewers of the study have an idea of how common it is for surgeons to check $v_2$ in addition to $v_1$, then they would be able to decide whether $\widehat\epsilon = 5 \%$ is large or small. In the supplementary material, we consider varying $\delta$, the parameter governing the severity of the unmeasured confounding. For instance, if $\delta = 0.5$ is thought to be reasonable, robustness would increase to $11.00\%$ ($95\% \text{ CI}=[3.84\%, 18.16\%]$) under the $X$-mixture model. 

Finally, we refer the readers to \cite{lin1998assessing} and \cite{altonji2008using}, among others, for additional sensitivity analyses applied to this dataset. In particular, in the context of Cox proportional hazard regression, and under certain simplifying assumptions, \cite{lin1998assessing} derive that a confidence interval for the relative hazard of death would include 1 as long as the prevalence of a binary unmeasured confounder is at least 10\% greater in the group that underwent RHC than in the control group. Using a probit model of mortality at day 90, \cite{altonji2008using} show that the observed positive association between mortality and RHC usage could be ``explained away" if the correlation between the unmeasured factors determining RHC usage and mortality is approximately 0.15. In addition, in Section \ref{section:extra_data_analysis} of the supplementary material, we apply the sensitivity analysis designed for linear models proposed in \cite{cinelli2020making}. We find that an unmeasured confounder that explains 4.2+\% of the variance in mortality not captured by RHC usage and the measured covariates and 4.2+\% of the variance in RHC usage not captured by the measured covariates would be sufficient to drive the observed effect ($\approx -0.04$) to zero. Notice that these approaches are designed for specific models used in the primary analysis, whereas our framework is agnostic regarding modeling choices. Further, they assume that the treatment-outcome association may be confounded for every unit, while our sensitivity model captures departures from such homogeneity by allowing the treatment-outcome association to be unconfounded for an unknown subgroup of units.  
\begin{figure*}[h!]
	\subfloat[$X$-mixture model ($S\ind (Y, A) \mid \Xb$) ]{%
		\includegraphics[width=0.5\textwidth]{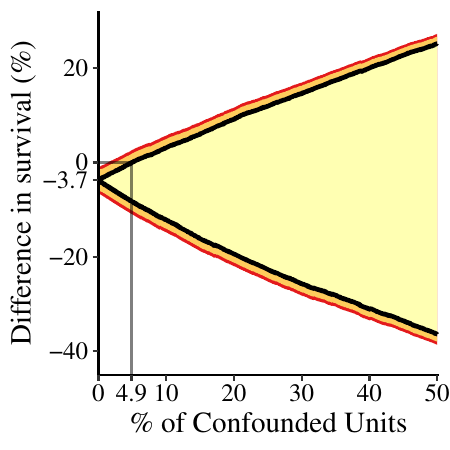}}
	\hspace{\fill}
	\subfloat[$XA$-mixture model ($S\ind Y \mid (A, \Xb)$)]{%
		\includegraphics[width=0.5\textwidth]{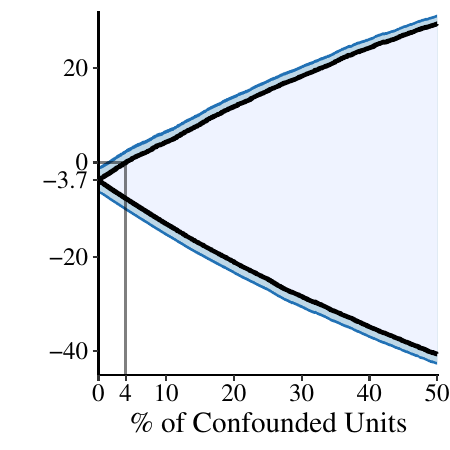}}\\
	\caption{\label{connors_results} Estimated bounds on the Average Treatment Effect as a function of the proportion of confounded units $\epsilon$ assuming ``worst-case" $\delta = 1$, with pointwise \citep{imbensmanskyci} and uniform 95\% confidence bands. Curves under the $X$-mixture model and under the $XA$-mixture model are shown along with estimates of $\epsilon_0$ on the abscissa.}
\end{figure*}
\section{Discussion}\label{section:discussion}
In this paper, we propose a novel approach to sensitivity analysis in observational studies where the sensitivity parameter is the proportion of unmeasured confounding. A strength of our model is that it captures a rich form of unmeasured confounding heterogeneity. While even richer models may allow for a more flexible characterization of confounding heterogeneity, we believe our approach strikes a nice balance between complexity and transparency. In fact, it captures heterogeneity with just one, intuitive sensitivity parameter:  an unknown fraction $\epsilon$ of the units can be arbitrarily confounded while the rest are not. The model is general enough to cover some relaxations to the no-unmeasured-confounding assumption already proposed in the literature. As $\epsilon$ is varied, lower and upper bounds on the ATE are derived under certain assumptions on the distribution of the confounded units. The parameter $\epsilon$ is interpretable and yields a natural one-number summary of a study's robustness to unmeasured confounding, namely the minimal proportion of confounding such that the bounds on the ATE contain zero. We provide sufficient conditions to construct both pointwise and uniform confidence bands around the curves tracing the lower and upper bounds on the ATE as a function of $\epsilon$. We also describe the asymptotic normality of a $Z$-estimator of $\epsilon_0$; we propose reporting an estimate of $\epsilon_0$ together with a Wald-type confidence interval when discussing results from an observational study.

Several questions remain unanswered and could be the subject of future research. First, bounding the ATE under no restrictions on the distribution of the confounded units is currently computationally intractable. Therefore, the discovery of a clever way to compute the bounds in this setting would generalize the current version of our model. Second, generalizing the approach of \cite{imbensmanskyci} to construct uniform confidence bands trapping the true ATE $\psi$, rather than the identification region $[\psi_l(\epsilon), \psi_u(\epsilon)] $, would allow far more precise inference. Lastly, extensions to our model other than the one considered in Appendix \ref{appendix:extra_extension} would likely lead to a richer set of sensitivity models, ultimately allowing the user to gauge the effects of departures from the no-unmeasured-confounding assumption in more nuanced ways. For example, it would be interesting to extend our sensitivity model to accommodate time-varying or continuous exposures, as well as to explore the possibility of tighter bounds by employing specific sensitivity analysis models to the confounded fraction of the sample. 
\section{Acknowledgments}
The authors thank Sivaraman Balakrishnan, Colin Fogarty, Marshall Joffe, Alan Mishler, Pratik Patil and members of the Causal Group at Carnegie Mellon University for helpful discussions. Edward Kennedy gratefully acknowledges financial support from NSF Grant DMS1810979.
\bibliographystyle{plainnat}
\bibliography{references}

\appendix
\addcontentsline{toc}{section}{Appendices}
\section{Proof of Lemma 1}
Notice that, because $A \ind Y^a \mid  \Xb, S = 1$, we have
\begin{align*}
\E\{( Y^1 - Y^0)S\} = \E[S\{\E( Y \mid  A = 1, \Xb, S = 1) - \E(Y \mid  A = 0, \Xb, S = 1)\}]
\end{align*}
and, by the consistency assumption, it holds that
\begin{align*}
\E\left\{ ( Y^1 - Y^0)(1-S) \right\} & = \E\left[ (1-S)\left\{(Y - Y^0)A + (Y^1 - Y) (1-A)\right\} \right] \\
& = \E((1-S)[\{Y - \lambda_{1-A}(\Xb)\}(2A - 1)])
\end{align*} 
Therefore, we conclude that
\begin{align*}
\psi = \E((1-S)[\{Y - \lambda_{1-A}(\Xb)\}(2A - 1)] + S\{\E( Y \mid  A = 1, \Xb, S = 1) - \E(Y \mid  A = 0, \Xb, S = 1)\})
\end{align*} 
as desired.
\section{Proof of Theorem 1}\label{appendix:proof_theorem_xmodel}
Notice that \eqref{x_mixture_assumption} is equivalent to $S \ind A \mid  \Xb$ and $S \ind Y \mid  \Xb, A$. Then, under  \eqref{x_mixture_assumption}, we have that $\E\left(Y \mid  \Xb, A = a, S \right) = \mu_a(\Xb)$ and $\Pb\left(A = a \mid  \Xb, S\right) = \pi(a\mid \Xb)$. This means that the result in Lemma \ref{general_def_ate} simplifies to
\begin{align*}
\psi(S, \lambda_0, \lambda_1) = \E\left(\mu_1(\Xb) - \mu_0(\Xb) + (1 - S) \left[\pi(0\mid \Xb) \left\{\lambda_1(\Xb) - \mu_1(\Xb) \right\} - \pi(1\mid \Xb)\left\{\lambda_0(\Xb) - \mu_0(\Xb) \right\}\right] \right)
\end{align*}
The observed distribution $\Pb$ and the knowledge of $S$ places no restrictions on $\lambda_0(\Xb)$ and $\lambda_1(\Xb)$. Recalling that $\delta$ is chosen such that
\begin{align*}
L_a \equiv \delta\{y_{\text{min}} - \mu_a(\Xb) \} \leq \lambda_a(\Xb) - \mu_a(\Xb) \leq \delta\{y_{\text{max}} - \mu_a(\Xb) \} \equiv U_a \text{ with prob. 1}
\end{align*}
for $a \in \{0, 1\}$, we have that
\begin{align*}
\E \left\{ \mu_1(\Xb) - \mu_0(\Xb) + (1-S) g(\etab) \right\} - \epsilon\delta(y_{\text{max}}-y_{\text{min}}) \leq \psi(S, \lambda_0, \lambda_1) \leq \E \left\{ \mu_1(\Xb) - \mu_0(\Xb) + (1-S) g(\etab) \right\}
\end{align*}
where $g(\etab) = \pi(0\mid\Xb) U_1 - \pi(1\mid\Xb)L_0$. These bounds are sharp for any given $S$. 

Next, notice that $g(\etab) : \mathcal{X}^p \to \R$ and $\Pb(S = 0) = \epsilon$. Thus, by Proposition 4 in \cite{horowitz1995identification}, it holds that $\psi \in [\psi_l(\epsilon), \psi_u(\epsilon)]$ where
\begin{align*}
\psi_l(\epsilon) & = \E\left[ \mu_1(\Xb) - \mu_0(\Xb) + \one\left\{ g(\etab) \leq q_{\epsilon} \right\} g(\etab) \right] - \epsilon\delta(y_{\text{max}}-y_{\text{min}}) \\
\psi_u(\epsilon) & = \E\left[ \mu_1(\Xb) - \mu_0(\Xb) + \one\left\{ g(\etab) > q_{ 1 - \epsilon }\right\} g(\etab) \right]
\end{align*}
and these bounds are sharp.
\section{Bounds in $XA$-mixture model}\label{appendix:xa-model}
The restriction in \eqref{x_mixture_assumption} can easily be weakened to 
\begin{align*}
S \ind Y \mid \Xb, A \tag{A2} \label{xa_mixture_assumption}
\end{align*}
Under \eqref{xa_mixture_assumption}, it still holds that $\E(Y \mid \Xb, A = a, S) = \mu_a(X)$, but $\pi(a \mid \Xb, S = 1)$ does not equal  $\pi(a \mid \Xb, S = 0)$ necessarily. Therefore, the result in Lemma \ref{general_def_ate} simplifies only to
\begin{align*}
\psi(S, \lambda_0, \lambda_1) = \E(\mu_1(\Xb) - \mu_0(\Xb) + (1-S)[(1-A)\{\lambda_1(\Xb) - \mu_1(\Xb)\} - A\{\mu_0(\Xb) - \lambda_0(\Xb)\}])
\end{align*}
where $\lambda_a(\Xb) = \E( Y^a \mid A = 1 - a, \Xb, S = 0)$. Following the same line of reasoning as in the proof of Theorem \ref{theorem_xmodel}, under consistency and positivity, sharp bounds on $\psi$ are:
\begin{align*}
\psi_l(\epsilon) & = \E[\mu_1(\Xb) - \mu_0(\Xb) + \one\{g(A, \etab) \leq q_{\epsilon} \}g(A, \etab)] - \epsilon \delta (y_{\text{min}} - y_{\text{max}})\\
\psi_u(\epsilon) & = \E[\mu_1(\Xb) - \mu_0(\Xb) + \one\{g(A, \etab) > q_{1-\epsilon}\}g(A, \etab)]
\end{align*}  
where  $g(A, \etab) = (1-A)U_1 - AL_0$, $q_{\tau}$ is the $\tau$-quantile of $g(A, \etab)$ and $\delta$ is chosen such that 
\begin{align*}
L_a \equiv \delta\{y_{\text{min}} - \mu_a(\Xb) \} \leq \lambda_a(\Xb) - \mu_a(\Xb) \leq \delta\{y_{\text{max}} - \mu_a(\Xb) \} \equiv U_a \text{ with prob. 1.}
\end{align*}
with $y_{\min}$ and $y_{\max}$ finite. The following lemma shows that the bounds assuming $S \ind Y \mid \Xb, A$ are at least as wide as those assuming $S \ind (Y, A) \mid \Xb$. 
\medskip
\begin{lemma}\label{lemma_expected_shortfall}
	Let $X, A$ be two random variables and let $\pi(X) = \E \left(A \mid X\right)$. Consider the functions:
	\begin{align*}
	g_1(a, x) = af(x) \quad \text{ and } \quad g_2(x) = \pi(x)f(x)
	\end{align*}
	for a measurable function $f$. Then, it holds that
	\begin{align}
	\E\left[ g_1(A, X) \one \left\{ g_1(A, X) \leq q_{1\tau} \right\} \right] & \leq \E\left[ g_2(X) \one \left\{ g_2(X) \leq q_{2\tau} \right\} \right] \nonumber \\
	\E\left[ g_1(A, X) \one \left\{ g_1(A, X) > q_{1\tau} \right\} \right] &  \geq \E\left[ g_2(X) \one \left\{ g_2(X) > q_{2\tau} \right\} \right] \label{eq_expected_shortfall_ub}
	\end{align}
	where $q_{i\tau}$ is the $\tau$-quantile of $g_i(\cdot)$.
\end{lemma}
\begin{proof}
	This lemma is essentially a restatement of the subadditivity property of expected shortfall \citep{acerbi2002coherence}. It is sufficient to note that
	\begin{align*}
	\E\left[ g_2(X) \one \left\{ g_2(X) \leq q_{2\tau} \right\} \right]  = \E\left[ g_1(A, X) \one \left\{ g_2(X) \leq q_{2\tau} \right\} \right]
	\end{align*}
	and that
	\begin{align*}
	\E\left( g_1(A, X) \left[\one \left\{ g_2(X) \leq q_{2\tau} \right\} - \one \left\{ g_1(A, X) \leq q_{1\tau} \right\} \right] \right) & \geq q_{1\tau} \E\left[\one \left\{ g_2(X) \leq q_{2\tau} \right\} - \one \left\{ g_1(A, X) \leq q_{1\tau} \right\} \right] \\
	& = q_{1\tau}(\tau - \tau) \\
	& = 0 
	\end{align*}
	where the inequality follows because
	\begin{align*}
	\begin{cases}
	\one \left\{ g_2(X) \leq q_{2\tau} \right\} - \one \left\{ g_1(A, X) \leq q_{1\tau} \right\} \leq 0 & \text{ if } g_1(A, X) \leq q_{1\tau} \\
	\one \left\{ g_2(X) \leq q_{2\tau} \right\} - \one \left\{ g_1(A, X) \leq q_{1\tau} \right\} \geq 0 & \text{ if } g_1(A, X) > q_{1\tau}
	\end{cases}
	\end{align*}
	Inequality \eqref{eq_expected_shortfall_ub} follows by rearranging:
	\begin{align*}
	\E\left[ g_1(A, X) \one \left\{ g_1(A, X) > q_{1\tau} \right\} \right] &  =  \E\left(g_1(A, X) \left[1-\one \left\{ g_1(A, X) \leq q_{1\tau} \right\}\right] \right) \\
	\E\left[ g_2(X) \one \left\{ g_2(X) > q_{2\tau} \right\} \right] &  =  \E\left( g_1(A, X) \left[1-\one \left\{ g_2(X) \leq q_{2\tau} \right\}\right] \right)
	\end{align*}
	so that
	\begin{align*}
	& \E\left( g_1(A, X) \left[\one \left\{ g_2(X) > q_{2\tau} \right\} - \one \left\{ g_1(A, X) > q_{1\tau} \right\} \right] \right) \\
	& \quad = \E\left( g_1(A, X) \left[\one \left\{ g_1(A, X) \leq q_{1\tau} \right\} - \one \left\{ g_2(X) \leq q_{2\tau} \right\} \right] \right) \\
	& \quad \leq 0
	\end{align*}
	as desired.
\end{proof}
From Lemma \ref{lemma_expected_shortfall} we conclude that the lower bound (upper bound) under $S \ind (Y, A) \mid \Xb$ is greater (smaller) than that under $S \ind Y \mid A , \Xb$. 
\section{Extensions} \label{appendix:extra_extension}
In this section, we discuss one possible extension to our model, though we note that others are possible. The impact of unmeasured confounding $U$ can be controlled by linking the true, unidentifiable propensity score $\Pb(A = a \mid \Xb, U, S = 0)$ to the estimable ``pseudo-propensity score" $\pi(a \mid \Xb)$ via a sensitivity model of choice. For example, as proposed in \cite{zhao2017sensitivity}, an extension to Rosenbaum's framework to non-matched data can be formulated by noting that, under consistency and positivity,
\begin{align}\label{zhao_starting_point}
\E \left(Y^a\right) = \E \left\{ \frac{Y \one \left( A = a \right)S}{\Pb(A = a \mid \Xb, S = 1, Y^a)} \right\} + \E \left\{ \frac{Y \one \left( A = a \right)(1-S)}{\Pb(A = a \mid \Xb, S = 0, Y^a)} \right\}
\end{align} 
and thus we can simply take the unobserved confounder $U$ to be one of the potential outcomes. Next, notice that $\Pb(A = a \mid \Xb, S = 1, Y^a) = \pi(a \mid \Xb)$ under Assumption \eqref{x_mixture_assumption} ($S \ind (Y, A) \mid \Xb$), so that \eqref{zhao_starting_point} simplifies to
\begin{align*}
\E \left(Y^a\right) & = \E \left\{ \frac{Y \one \left(A = a \right)S}{\pi(a\mid \Xb)} \right\} + \E \left\{ \frac{Y \one \left( A = a \right)(1-S)}{\Pb(A = a \mid \Xb, S = 0, Y^a)} \right\} 
\end{align*}
Let $\pi_a(\xb, y) = \Pb(A = a \mid \Xb = \xb, S = 0, Y^a = y)$. Noting that $\Pb (A = a \mid \Xb, S = 0) = \pi(a \mid \Xb)$ under Assumption \eqref{x_mixture_assumption}, the impact of unmeasured confounding can be governed by requiring $\pi_a(\xb, y)$ to be an element of the following sensitivity model
\begin{align}\label{zhao_model}
\mathcal{E}\left(\Lambda \right) = \left\{ \Lambda^{-1} \leq \text{OR} \left\{\pi_a(\mathbf{x}, y), \pi(a\mid \mathbf{x}) \right\} \leq \Lambda, \text{ for all } \mathbf{x} \in \mathcal{X}, \ y \in [0, 1], \ a \in \{0, 1\} \right\}
\end{align}
where $\Lambda \geq 1$ and $\Lambda = 1$ corresponds to the unconfounded case. Model \eqref{zhao_model} can be conveniently reformulated on the logit scale. Let
\begin{align*}
& g(a\mid \mathbf{x}) = \text{logit}\{\pi(a\mid \mathbf{x})\}, \quad g_a(\mathbf{x}, y) = \text{logit}\{\pi_a(\mathbf{x}, y)\} \\
& h(\mathbf{x}, y) =  g(a\mid \mathbf{x}) - g_a(\mathbf{x}, y), \quad \pi^{(h)} (\mathbf{x}, y) = \left[1 + \exp\left\{ h(\mathbf{x}, y) - g(a\mid \mathbf{x}) \right\}\right] ^{-1}
\end{align*}
and write
\begin{align}\label{zhao_model_logit}
\mathcal{E}\left(\Lambda \right) = \left\{ \pi^{(h)} (\mathbf{x}, y) : \ h \in \mathcal{H}(\Lambda) \right\},  \text{ where } \mathcal{H}(\Lambda) = \left\{h: \mathcal{X} \times [0, 1] \to \R \text{ and } \|  h \| _\infty \leq \log\Lambda \right\}
\end{align}
From \eqref{zhao_model_logit}, we rewrite $\E\left(Y^a\right)$ as
\begin{align}\label{zhao_end_point}
\E\left(Y^a\right) & = \E\left(\frac{S Y\one\left(A = a\right)}{\pi(a\mid \Xb)} + (1 - S)Y\one\left(A = a\right)\left[1 + \exp\left\{h(\Xb, Y)\right\} \exp\left\{- g(a\mid \Xb)\right\} \right] \right)
\end{align}
where $\exp\left\{h(\Xb, Y)\right\} \in [\Lambda^{-1}, \Lambda ]$. Bounds on $\psi$ can then be computed following the same line of reasoning as in Theorem \ref{theorem_xmodel}, where $\exp\left\{h(\Xb, Y)\right\}$ takes the role of $\lambda_a(\Xb)$. Convergence statements for estimators of \eqref{zhao_end_point} can be derived using standard arguments for convergence of inverse propensity score-weighted estimators together with the arguments made in proving Theorem \ref{gp_convergence}. However, we expect the conditions for $\sqrt{n}$-consistency and asymptotic normality to be stronger than those assumed in Theorem \ref{gp_convergence}. Moreover, note that, if $\Pb(S = 1) = 0$, as in \cite{zhao2017sensitivity}, expression \eqref{zhao_end_point} can be bounded and estimated via a stabilized IPW (SIPW) and a suitable linear program. In our model, because $\Pb(S = 1) \geq 0$, optimization of a SIPW is harder due to the integer nature of $S$ and beyond the scope of this paper. 
\section{Technical Proofs}
\subsection{Proof of Theorem \ref{gp_convergence}}\label{proof_gp_convergence}
Before proceeding with the proof of Theorem \ref{gp_convergence}, we report a lemma used below. It can be found in \cite{kennedy2018sharp} (Lemma 1) or in the proof of Lemma 2 in \cite{van2014targeted}.
\medskip
\begin{lemma}\label{ind_lemma}
	Let $\widehat{f}$ and $f$ take any real values. Then
	\begin{align*}
	|\one(\widehat{f} > 0) - \one (f > 0)| \ \leq \one ( |f|  \leq |\widehat{f} - f|)
	\end{align*}
\end{lemma}
\begin{proof}
	This follows since
	\begin{align*}
	|\one(\widehat{f} > 0) - \one (f > 0)| \ = \one(\widehat{f}, f \text{ have opposite sign})
	\end{align*}
	and if $\widehat{f}$ and $f$ have opposite sign then
	\begin{align*}
	|\widehat{f}| \ + \ |f| = |\widehat{f} - f|
	\end{align*}
	which implies that $| f | \ \leq |\widehat{f} - f |$. Therefore, whenever $|  \one(\widehat{f} > 0) - \one(f > 0)| \ =1$, it must also be the case that $\one(|f| \ \leq |\widehat{f} - f|) = 1$, which yields the result.
\end{proof}

The proof of Theorem \ref{gp_convergence} is similar to that of Theorem 3 in \cite{kennedy2018nonparametric}, with the main difference being that the influence function of the estimator proposed is not a smooth function of the sensitivity parameter $\epsilon$. Fortunately, we can exploit the fact that the bounds are monotone in $\epsilon$ to establish convergence to a Gaussian process. We prove the result for the upper bound, as the case for the lower bound follows analogously. We also proceed by assuming $Y$ is bounded in $[0, 1]$. 

Let $\|  f \| _\Ec \ = \sup_{\epsilon \in \Ec}|f(\epsilon)|  $ denote the supremum norm over $\Ec \subseteq [0, 1]$, a known interval. Let $\lambda_{1-\epsilon}$ be shorthand notation for $\one\left\{ g(\etab) > q_{1-\epsilon} \right\}$. Similarly, let $\tau$ and $\nu$ be shorthand notations for the uncentered influence functions of $\E\left\{g(\etab)\right\}$ and $\E\{\mu_1(\Xb) - \mu_0(\Xb)\}$ respectively, so that
\begin{align*}
& \tau = \frac{(1-2A)\left\{Y - \mu_A(\Xb)\right\}}{\pi(A\mid \Xb) / \pi(1-A\mid \Xb)}+ A \mu_0(\Xb) + (1-A)\left(1 - \mu_1(\Xb)\right) \\
& \nu = \frac{(2A - 1)\left\{Y - \mu_A(\Xb) \right\}}{\pi(A\mid \Xb)} + \mu_1(\Xb) - \mu_0(\Xb)
\end{align*}
Define the following processes:
\begin{align*}
\widehat{\Psi}_n(\epsilon) & = \sqrt{n} \{\widehat{\psi}_u(\epsilon) - \psi_u(\epsilon)\} / \widehat{\sigma}_u(\epsilon) \\
\tilde{\Psi}_n(\epsilon) & = \sqrt{n}\{\widehat{\psi}_u(\epsilon) - \psi_u(\epsilon)\} / \sigma_u(\epsilon) \\
\Psi_n(\epsilon) & = \Gn([ \varphi_u(\Ob; \etab, q_{1-\epsilon}) - \lambda_{1-\epsilon} q_{1-\epsilon} -  \{\psi_u(\epsilon) - \epsilon q_{1-\epsilon}\} ] / \sigma_u(\epsilon)) \\
& = \Gn([ \overline{\varphi}_u(\Ob; \etab, q_{1-\epsilon}) -  \{\psi_u(\epsilon) - \epsilon q_{1-\epsilon} \}] / \sigma_u(\epsilon)) \\
& = \Gn\{\phi_u(\Ob; \etab, q_{1-\epsilon})\}
\end{align*}
where $\overline{\varphi}_u(\Ob; \etab, q_{1-\epsilon})=~{\nu + \lambda_{1-\epsilon}(\tau - q_{1-\epsilon})}$ and $\Gn (\cdot) = \sqrt{n}(\Pn - \Pb)$ denotes the empirical process on the full sample.

We also let $\Gb(\cdot)$ denote the mean-zero Gaussian process with covariance $\E\left\{\phi_u(\Ob; \etab, q_{1-\epsilon_1})\phi_u(\Ob; \etab, q_{1-\epsilon_2})\right\}$. We will show that $\Psi_n(\cdot) \indist \Gb(\cdot)$ in $\ell^{\infty}(\Ec)$ and that $\|  \widehat{\Psi}_n - \Psi_n\| _\Ec = o_\Pb(1)$. 

To show that $\Psi_n(\cdot) \indist \Gb(\cdot)$ in $\ell^{\infty}(\Ec)$, notice that $\overline{\varphi}_u(\cdot; \etab, q_{1-\epsilon}): \Ec \to [-M, M]$, for some $M < \infty$, consists of a sum of a bounded, constant function plus a product of two monotone functions. Specifically, consider $s(\cdot; \etab, \epsilon): \Ec \mapsto [-S, S]$, defined as $s(\cdot; \etab, \epsilon) = \nu$, $f(\cdot; \etab, \epsilon) : \Ec \mapsto \{0, 1\}$, defined as $f(\cdot; \etab, \epsilon) = \lambda_{1-\epsilon}$, and $h(\cdot; \etab, \epsilon) : \Ec \mapsto [-H, H]$, defined as $h(\cdot; \etab, \epsilon) = \tau - q_{1-\epsilon}$. Then, $\overline{\varphi}_u(\cdot; \etab, q_{1-\epsilon}) = s(\cdot; \etab, \epsilon) + f(\cdot; \etab, \epsilon) h(\cdot; \etab, \epsilon)$. The fact that $s(\cdot; \etab, \epsilon)$ and $h(\cdot; \etab, \epsilon)$ are uniformly bounded follows by the assumptions that $\Pb\{t \leq \pi(a\mid \Xb) \leq 1-t\} =1$, for some $t > 0$ and $a \in \{0,1\}$, and that the outcome $Y$ is bounded. 

Then we define the class $\mathcal{F}_\eta$ where $\overline{\varphi}_u(\cdot; \etab, q_{1-\epsilon})$ takes value in
\begin{align*}
\mathcal{F}_\eta = \left\{ \nu + \lambda_{1-\epsilon}(\tau - q_{1-\epsilon}) : \epsilon \in \Ec \right\}.
\end{align*} 
$\mathcal{F}_\eta$ is contained in the sum of $\mathcal{F}_{\eta, 0}$ and the pairwise product $\mathcal{F}_{\eta, 1} \cdot \mathcal{F}_{\eta, 2}$, where  $\mathcal{F}_{\eta, 0} = \{\nu: \epsilon \in \Ec\}$ (constant function class), $\mathcal{F}_{\eta, 1} = \{\lambda_{1-\epsilon}: \epsilon \in \Ec \}$ and $\mathcal{F}_{\eta, 2}~=~\{ \tau - q_{1-\epsilon}:~\epsilon \in \Ec \}$. 

By, for example, Theorem 2.7.5 in \cite{van1996weak}, the class of bounded monotone functions possesses a finite bracketing integral, and in particular, for $w \in \{0, 1,2\}$:
\begin{align*}
\log N_{[]}\left( \delta, \mathcal{F}_{\eta, w}, L_2(\Pb) \right) \lesssim \frac{1}{\delta}
\end{align*}
Furthermore, because $\mathcal{F}_{\eta, 0}$, $\mathcal{F}_{\eta, 1}$ and $\mathcal{F}_{\eta, 2}$ are uniformly bounded:
\begin{align*}
\log N_{[]}\left( \delta, \mathcal{F}_{\eta}, L_2(\Pb) \right) \lesssim 3 \log N_{[]}\left( \frac{\delta}{2}, \mathcal{F}_{\eta, 1}, L_2(\Pb) \right) \lesssim \frac{1}{\delta}
\end{align*}
by, for instance, Lemma 9.24 in \cite{kosorok2008introduction}. Thus, by for example Theorem 19.5 in \cite{van2000asymptotic}, $\mathcal{F}_{\eta}$ is Donsker.

Next, we prove the statement that $\| \widehat{\Psi}_n - \Psi_n \| _\Ec \ = o_\Pb(1)$. First, we notice that
\begin{align*}
\|  \widehat{\Psi}_n - \Psi_n\|_\Ec  & = \|  (\tilde{\Psi}_n - \Psi_n ) \sigma_u / \widehat{\sigma}_u + \Psi_n \left(\sigma_u - \widehat{\sigma}_u \right) / \widehat{\sigma}_u \|_\Ec  \\
& \leq \|  \tilde{\Psi}_n - \Psi_n \| _\Ec \|  \sigma_u/\widehat{\sigma}_u \| _\Ec \ + \ \|  \sigma_u / \widehat{\sigma}_u - 1 \| _\Ec \|  \Psi_n \| _\Ec \\
& \lesssim \|  \tilde{\Psi}_n - \Psi_n \| _\Ec \ + \ o_\Pb(1)
\end{align*}
where the last inequality follows because $\|  \widehat{\sigma}_u / \sigma_u - 1\| _\Ec \ = o_\Pb(1)$ by assumption and $\|  \Psi_n\| _\Ec \ = O_\Pb(1)$ by, for example, Theorem 2.14.2 in \cite{van1996weak} since $\mathcal{F}_\eta$ possesses a finite bracketing integral.

Let $N = n / B$ be the number of samples in any group $k = 1, \ldots, B$, and denote the empirical process over group $k$ units by $\Gn^k = \sqrt{N}(\Pn^k - \Pb)$. Then, we have
\begin{align*}
\tilde{\Psi}_n(\epsilon) - \Psi_n(\epsilon) & = \frac{\sqrt{n}}{\sigma_u(\epsilon)} \{\widehat{\psi}_u(\epsilon) - \psi_u(\epsilon)\} - \Gn\{\tilde{\varphi}_u(\Ob; \etab, q_{1-\epsilon})\} \\
& = \frac{\sqrt{n}}{B \sigma_u(\epsilon)}\sum_{k=1}^B \left[ \Pn^k\{ \varphi_u(\Ob; \widehat{\etab}_{-k}, \widehat{q}_{-k, 1-\epsilon})\} - \psi_u(\epsilon)  - (\Pn - \Pb)\{\overline{\varphi}(\Ob; \etab, q_{1-\epsilon})\} \right] \\
& = \frac{\sqrt{n}}{B \sigma_u(\epsilon)}\sum_{k=1}^B \left[ \Pn^k\{ \varphi_u(\Ob; \widehat{\etab}_{-k}, \widehat{q}_{-k, 1-\epsilon}) - \varphi_u(\Ob; \etab, q_{1-\epsilon})\} + \left(\Pn - \Pb\right)(\lambda_{1-\epsilon} q_{1-\epsilon}) \right] \\
& = \frac{\sqrt{n}}{B\sigma_u(\epsilon)}\sum_{k=1}^B\left[ \frac{1}{\sqrt{N}} \Gn^k\{ \varphi_u(\Ob; \widehat{\etab}_{-k}, \widehat{q}_{-k, 1-\epsilon}) - \varphi_u(\Ob; \etab, q_{1-\epsilon})\} \right. \\
& \hphantom{=\frac{\sqrt{n}}{B\sigma_u(\epsilon)}\sum_{k=1}^B\left[ \right.} \left. \vphantom{\frac{1}{\sqrt{N}}} + \Pb\{ \varphi_u(\Ob; \widehat{\etab}_{-k}, \widehat{q}_{-k, 1-\epsilon}) - \varphi_u(\Ob; \etab, q_{1-\epsilon})\} + (\Pn - \Pb)(\lambda_{1-\epsilon} q_{1-\epsilon} ) \right]
\end{align*}
where we used the facts that 
\begin{align*}
\psi_u(\epsilon) = \Pb\left\{ \varphi_u(\Ob; \etab, q_{1-\epsilon}) \right\}\quad \text{and} \quad  \sum_{k=1}^B \Pn^k\left\{\varphi_u(\Ob; \etab, q_{1-\epsilon})\right\} = \sum_{k=1}^B \Pn \left\{\varphi_u(\Ob; \etab, q_{1-\epsilon})\right\}
\end{align*}
The term $\Pb\left\{ \varphi_u(\Ob; \widehat{\etab}_{-k}, \widehat{q}_{-k, 1-\epsilon}) - \varphi_u(\Ob; \etab, q_{1-\epsilon}) \right\}$ can be decomposed in
\begin{align*}
\Pb\{ \varphi_u(\Ob; \widehat{\etab}_{-k}, \widehat{q}_{-k, 1-\epsilon}) - \varphi_u(\Ob; \etab, q_{1-\epsilon})\} & = \Pb\{\widehat{\nu}_{-k} - \nu + \widehat{\lambda}_{-k, 1-\epsilon} \left(\widehat{\tau}_{-k} - \tau \right) + (\tau - q_{1-\epsilon})(\widehat{\lambda}_{-k, 1-\epsilon} - \lambda_{1-\epsilon}) \\
& \hphantom{= \Pb\{ } \ + q_{1-\epsilon}(\widehat{\lambda}_{-k, 1-\epsilon} - \lambda_{1-\epsilon}) \}
\end{align*}
Notice that $\epsilon + o_\Pb(n^{-1/2}) = \Pn^k(\widehat{\lambda}_{-k, 1-\epsilon}) = \Pb \left(\lambda_{1-\epsilon} \right)$, so that
\begin{align*}
o_\Pb(n^{-1/2}) & = \sum_{k = 1}^B \Pn^k(\widehat{\lambda}_{-k, 1-\epsilon}) - \Pb(\lambda_{1-\epsilon}) \\
& = \sum_{k = 1}^B \left(\Pn^k - \Pb\right)(\widehat{\lambda}_{-k, 1-\epsilon} - \lambda_{1-\epsilon}) + \Pb(\widehat{\lambda}_{-k, 1-\epsilon} - \lambda_{1-\epsilon}) + \left(\Pn - \Pb\right)(\lambda_{1-\epsilon})
\end{align*}
where we used again the fact that $\sum_{k = 1}^B \Pn^k (\lambda_{1-\epsilon}) = \sum_{k = 1}^B \Pn (\lambda_{1-\epsilon})$. Thus, we have that
\begin{align*}
\sum_{k = 1}^B \Pb\{ q_{1-\epsilon}( \widehat{\lambda}_{-k, 1-\epsilon} - \lambda_{1-\epsilon})\} & = - \sum_{k = 1}^B (\Pn^k - \Pb)\{q_{1-\epsilon}( \widehat{\lambda}_{-k, 1-\epsilon} - \lambda_{1-\epsilon})\} - \left(\Pn - \Pb\right)(q_{1-\epsilon}\lambda_{1-\epsilon}) + o_\Pb(n^{-1/2})
\end{align*}
Therefore, we rewrite $\tilde{\Psi}_n(\epsilon) - \Psi_n(\epsilon)$ as
\begin{align*}
\tilde{\Psi}_n(\epsilon) - \Psi_n(\epsilon) & = \frac{\sqrt{n}}{B\sigma_u(\epsilon)}\sum_{k=1}^B\left( \frac{1}{\sqrt{N}} \Gn^k\{\varphi_u(\Ob; \widehat{\etab}_{-k}, \widehat{q}_{-k, 1-\epsilon}) - q_{1-\epsilon} \widehat{\lambda}_{-k, 1-\epsilon} - \overline{\varphi}_u(\Ob; \etab, q_{1-\epsilon})\} \right. \\
& \hphantom{= \frac{\sqrt{n}}{B\sigma_u(\epsilon)}\sum_{k=1}^B\left[ \right.} + \left. \vphantom{\frac{1}{\sqrt{N}}} \Pb\{\widehat{\nu}_{-k} - \nu + \widehat{\lambda}_{-k, 1-\epsilon} (\widehat{\tau}_{-k} - \tau) + (\tau - q_{1-\epsilon})(\widehat{\lambda}_{-k, 1-\epsilon} - \lambda_{1-\epsilon})\} \right) \\
& \equiv B_{n, 1}(\epsilon) + B_{n, 2}(\epsilon) + o_\Pb(1)
\end{align*}

Next, we show that $\left\| B_{n, 1}\right\| _\Ec = o_\Pb(1)$ and $\left\| B_{n, 2} \right\| _\Ec = o_\Pb(1)$, which completes the proof. 

For $B_{n,1}(\epsilon)$, notice that, because $B$ is fixed regardless of $n$, we have that
\begin{align*}
\left\| B_{n, 1}\right\| _\Ec & = \sup_{\epsilon \in \Ec} \left|  \frac{1}{\sqrt{B}\sigma_u(\epsilon)}\sum_{k=1}^B \Gn^k\{ \varphi_u(\Ob; \widehat{\etab}_{-k}, \widehat{q}_{-k, 1-\epsilon}) - q_{1-\epsilon} \widehat{\lambda}_{-k, 1-\epsilon} - \overline{\varphi}_u(\Ob; \etab, q_{1-\epsilon}) \} \right|  \\
& \lesssim \max_k \sup_{f \in \mathbf{\mathcal{F}}_n^k} |\Gn(f)| 
\end{align*}
where we define the class $\mathbf{\mathcal{F}}_n^k = \mathbf{\mathcal{F}}_{\widehat{\etab}_{-k}} - \mathbf{\mathcal{F}}_{\etab} $, where $\mathcal{F}_{\widehat{\etab}_{-k}} = \{ \widehat{\nu}_{-k} + \widehat{\lambda}_{-k, 1-\epsilon} (\widehat{\tau}_{-k} - q_{1-\epsilon}): \epsilon \in \Ec \}$ and $\mathcal{F}_\eta = \left\{ \overline{\varphi}_u(\cdot; \etab, \epsilon): \epsilon \in \Ec \right\}$ as above. Viewing $\widehat{\etab}_{-k}$ as fixed given the training data $D_0^k = \left\{\Ob_i : K_i \neq k \right\}$, by Theorem 2.14.2 in \cite{van1996weak}, we have that
\begin{align*}
\E\left\{ \sup_{f \in \mathbf{\mathcal{F}}_n^k} \left|  \Gn(f) \right|  \left| \right.  D_0^k \right\} \lesssim \left\|  F_n^k \right\|  \int_0^1 \sqrt{1 + \log N_{[]}\left(\delta \left\| F_n^k\right\| , \mathcal{F}_n^k, L_2(\Pb) \right) d\delta}
\end{align*} 
where $F_n^k$ is an envelop of the class $\mathcal{F}_n^k$. Given the training data, the function class where $\widehat{\lambda}_{-k, 1-\epsilon}$ lives can be expressed as $\{ \one (u > q), q \in \mathcal{Q} \}$, where $\mathcal{Q}$ is the set of all quantile functions, which in this case is a subset of the class of all bounded, monotone functions because $g(\etab_{-k})$ is bounded for any $k$. Therefore, by the same line of argument as above, the class $\mathbf{\mathcal{F}}_n^k$ is contained in unions and products of classes of uniformly bounded, monotone functions. As such, it satisfies
\begin{align*}
\log N_{[]}\left(\delta \left\| F_n^k\right\| , \mathcal{F}_n^k, L_2(\Pb) \right) \lesssim \frac{1}{\delta \left\| F_n^k\right\|}
\end{align*}
If we take 
\begin{align*}
F_n^k(\mathbf{o}) = \sup_{\epsilon \in \Ec} | \varphi_u(\Ob; \widehat{\etab}_{-k}, \widehat{q}_{-k, 1-\epsilon}) - q_{1-\epsilon} \widehat{\lambda}_{-k, 1-\epsilon} - \overline{\varphi}_u(\mathbf{o}; \etab, q_{1-\epsilon})|
\end{align*}
then $\|F_n^k\|  = o_\Pb(1)$ by assumption. The bracketing integral is finite for any fixed $\etab$, but here $\mathcal{F}_n^k$ depends on $n$ through $\widehat{\etab}_{-k}$, hence concluding that the LHS is $o_\Pb(1)$ requires further analysis. Letting $C_n^k = \| F_n^k\| $, we have that
\begin{align*}
\left\|  F_n^k \right\|  \int_0^1 \sqrt{1 + \log N_{[]}\left(\delta \left\| F_n^k\right\| , \mathcal{F}_n^k, L_2(\Pb) \right) d\delta} & \lesssim C_n^k \int_0^1\sqrt{1 + \frac{1}{\delta C_n^k}} d\delta \\
& = \sqrt{C_n^k(C_n^k + 1)} + \frac{1}{2}\log \left\{ 1 + 2C_n^k\left(1 + \sqrt{ 1 + \frac{1}{C_n^k}}\right) \right\}
\end{align*}
which goes to zero as $C_n^k \to 0$. Hence, we conclude that $\sup_{f\in \mathcal{F}_n^k}|\Gn(f)|  = o_\Pb(1)$ for each $k$. Because $B$ is finite, this implies that $\| B_{n, 1}\| _\Ec \ = o_\Pb(1)$ as desired. 

For $B_{n,2}(\epsilon)$, first notice that
\begin{align*}
\Pb(\widehat{\nu}_{-k} - \nu) & \lesssim \Pb \left[\{\pi(1 \mid \Xb) - \widehat{\pi}(1\mid \Xb)\}\left\{\frac{\mu_1(\Xb) - \widehat{\mu}_1(\Xb)}{\widehat{\pi}(1\mid \Xb)} + \frac{\mu_0(\Xb) - \widehat{\mu}_0(\Xb)}{1-\widehat{\pi}(1\mid\Xb)}\right\} \right] \\
& \lesssim \left\|\widehat{\pi}(1\mid \Xb) - \pi(1\mid \Xb) \right\| \max_a \left\|  \widehat{\mu}_a(\Xb) - \mu_a(\Xb) \right\|
\end{align*}
by an application of the Cauchy-Schwartz inequality. 

Next, similar calculations yield
\begin{align*}
\sup_{\epsilon \in \Ec} \Pb\{\widehat{\lambda}_{-k, 1-\epsilon}(\widehat{\tau}_{-k} - \tau)\}  & \leq \Pb(|\widehat{\tau}_{-k} - \tau|) \\
& \lesssim \left\|  \widehat{\pi}(1\mid \mathbf{X}) - \pi(1\mid \mathbf{X}) \right\|  \left(\max_a \left\|  \widehat{\mu}_a(\mathbf{X}) - \mu_a(\mathbf{X}) \right\| \right)
\end{align*}
where the first inequality follows because $\sup_{\epsilon \in \Ec}|\widehat{\lambda}_{-k, 1-\epsilon}|  \ \leq 1$.

Finally, we have
\begin{align*}
\Pb \{(\tau - q_{1-\epsilon})(\widehat{\lambda}_{-k, 1-\epsilon} - \lambda_{1-\epsilon})\} 
& \lesssim \Pb[(\widehat{\lambda}_{1-\epsilon} - \lambda_{1-\epsilon})\{g(\etab) - q_{1-\epsilon}\}]   \\
& \leq \Pb[ \left| g\left(\etab\right) - q_{1-\epsilon} \right|  \left|  \one \left\{g(\widehat{\etab}) - \widehat{q}_{1-\epsilon} > 0 \right\} - \one \left\{g(\etab) - q_{1-\epsilon} > 0 \right\} \right|] \\
& \leq \Pb[\left|  g\left(\etab\right) - q_{1-\epsilon} \right|  \one \left\{ \left|  g\left(\etab\right) - q_{1-\epsilon} \right|  \leq \left| g(\etab) - g(\widehat{\etab}) \right|  + \left|  \widehat{q}_{1-\epsilon} - q_{1-\epsilon} \right|  \right\}] \\
& \lesssim ( \left\|  g(\widehat{\etab}) - g\left(\etab\right) \right\| _\infty + \left|  \widehat{q}_{1-\epsilon} - q_{1-\epsilon} \right|)^{1+\alpha}
\end{align*}
where the third inequality follows by Lemma \ref{ind_lemma} and the last inequality follows by the margin condition (assumption \eqref{margin_condition}).

Therefore, we have that
\begin{align*}
\frac{\left\| B_{n, 2} \right\|_\Ec}{\sqrt{n}}& \lesssim \left\|  \widehat{\pi}(1\mid \mathbf{X}) - \pi(1\mid \mathbf{X}) \right\| \max_a \left\|  \widehat{\mu}_a(\mathbf{X}) - \mu_a(\mathbf{X}) \right\| +  \left( \left\|  g(\widehat{\etab}) - g\left(\etab\right) \right\| _\infty + \sup_{\epsilon \in \Ec} \left|  \widehat{q}_{1-\epsilon} - q_{1-\epsilon} \right|  \right)^{1+\alpha}
\end{align*}
where the RHS is $o_\Pb(n^{-1/2})$ by assumption. 
\subsection{Construction of Uniform Confidence Bands} \label{section:uniform_bands}
In this section, we propose the construction of $1-\alpha$ confidence bands capturing $\psi$ uniformly in $\epsilon$. For any given $\epsilon$, confidence intervals for $\psi$ can be constructed in at least two ways. One way is to construct a confidence interval for the identification region $[\psi_l(\epsilon), \psi_u(\epsilon) ]$. Another way is to construct a confidence interval for $\psi$ directly \citep{imbensmanskyci, stoye2009more, vansteelandt2006ignorance}. The former approach yields a conservative confidence interval for $\psi$, particularly for larger values of $\epsilon$ for which the identification interval is wider. To see this, notice that, unless the length of the interval is of the same order as the sampling variability, the true parameter $\psi$ can be close to either the lower bound or the upper bound, but not to both. Thus, the confidence interval in regimes of large $\epsilon$ is practically one-sided. Here, we provide confidence bands for the identification region that are valid uniformly over $\epsilon$. These bands also serve as conservative uniform bands for the true $\psi$ curve. We also provide the code to construct bands covering just $\psi(\epsilon)$, as in \cite{imbensmanskyci}, that are valid pointwise. We leave the construction of bands covering just $\psi(\epsilon)$ that are valid uniformly over $\epsilon$ for future research.

Let sample analogues of the variance functions of the bounds at $\epsilon$ be
\begin{align*}
& \widehat{\sigma}^2_u(\epsilon) = \Pn ([ \varphi_u(\Ob; \widehat{\etab}_{-K}, \widehat{q}_{1-\epsilon, -K}) - \one\{g(\widehat{\etab}_{-K}) > \widehat{q}_{1-\epsilon, -K} \}\widehat{q}_{1-\epsilon, -K} - \widehat{\psi}_u(\epsilon) + \epsilon \widehat{q}_{1-\epsilon, -K}]^2) \\
& \widehat{\sigma}^2_l(\epsilon) = \Pn ([ \varphi_l(\Ob; \widehat{\etab}_{-K}, \widehat{q}_{\epsilon, -K}) - \one\{g(\widehat{\etab}_{-K}) \leq \widehat{q}_{\epsilon, -K} \}\widehat{q}_{\epsilon, -K} - \widehat{\psi}_l(\epsilon) + \epsilon \widehat{q}_{\epsilon, -K}]^2).
\end{align*}
To construct asymptotically valid $(1-\alpha)$-uniform bands of the form 
\begin{align}\label{uniform_bands}
\widehat{\text{CI}}(\epsilon; c_\alpha, d_\alpha) = \left[\widehat{\psi}_l(\epsilon) - c_\alpha \frac{\widehat{\sigma}_l(\epsilon)}{\sqrt{n}}, \widehat{\psi}_u(\epsilon) + d_\alpha \frac{\widehat{\sigma}_u(\epsilon)}{\sqrt{n}}\right],
\end{align}
we need to find the critical values $c_\alpha$ and $d_\alpha$ such that
\begin{align*} 
& \Pb \left[ \sup_{\epsilon \in \Ec} \left\{\frac{\widehat{\psi}_l(\epsilon) - \psi_l(\epsilon)}{\widehat{\sigma}_l(\epsilon)/ \sqrt{n}} \right\} \leq c_\alpha \text{ and } \sup_{\epsilon \in \Ec} \left\{ \frac{\psi_u(\epsilon) - \widehat{\psi}_u(\epsilon)}{\widehat{\sigma}_u(\epsilon)/ \sqrt{n}}  \right\} \leq d_\alpha \right] \geq 1 - \alpha + o(1)
\end{align*}
In particular, we propose choosing $c_\alpha$ and $d_\alpha$ such that
\begin{align}\label{calpha_def}
& \Pb \left[ \sup_{\epsilon \in \Ec} \left\{ \frac{\widehat{\psi}_l(\epsilon) - \psi_l(\epsilon)}{\widehat{\sigma}_l(\epsilon)/ \sqrt{n}} \right\} \leq c_\alpha \right] =  \Pb \left[ \sup_{\epsilon \in \Ec} \left\{ \frac{\psi_u(\epsilon) - \widehat{\psi}_u(\epsilon)}{\widehat{\sigma}_u(\epsilon)/ \sqrt{n}} \right\}  \leq d_\alpha \right] = 1-\frac{\alpha}{2} + o(1),
\end{align}
essentially allowing the lower (upper) bound estimate to be greater (smaller) than the true lower (upper) bound with probability equal to $\alpha/2$. In light of the result in Theorem \ref{gp_convergence},  $c_\alpha$ and $d_\alpha$ can be found by approximating the distribution of the supremum of the respective Gaussian processes. Similarly to \cite{kennedy2018nonparametric}, we use the multiplier bootstrap to approximate these distributions. A key advantage of this approximating method is its computational efficiency, as it does not require refitting the nuisance functions estimators. 

The following lemma asserts that, for $\xi$ and $\zeta$ iid Rademacher random variables, the suprema of the following multiplier processes
\begin{align*}
& \sqrt{n} \Pn(\zeta [\varphi_l(\Ob; \widehat{\etab}_{-K}, \widehat{q}_{\epsilon; - K}) - \one\{g(\widehat{\etab}_{-K}) \leq \widehat{q}_{\epsilon, -K} \} \widehat{q}_{\epsilon, -K} - \widehat{\psi}_l(\epsilon) + \epsilon \widehat{q}_{\epsilon; -K} ] / \widehat{\sigma}_l(\epsilon)) \\
& \sqrt{n} \Pn(\xi [\widehat{\psi}_u(\epsilon) - \epsilon \widehat{q}_{1-\epsilon; -K} - \varphi_u(\Ob; \widehat{\etab}_{-K}, \widehat{q}_{1-\epsilon; - K}) + \one\{g(\widehat{\etab}_{-K}) > \widehat{q}_{1-\epsilon, -K} \} \widehat{q}_{1-\epsilon, -K} ] / \widehat{\sigma}_u(\epsilon))
\end{align*}
are valid approximations to their counterparts in \eqref{calpha_def}.
\medskip
\begin{lemma}\label{mult_theorem}
	Conditional on the sample, let $\widehat{c}_\alpha$ and $\widehat{d}_\alpha$ denote the $(1-\alpha/2)$-quantiles of
	\begin{align*}
	& \sup_{\epsilon \in \Ec} \sqrt{n} \Pn(\zeta [\varphi_l(\Ob; \widehat{\etab}_{-K}, \widehat{q}_{\epsilon; - K}) - \one\{g(\widehat{\etab}_{-K}) \leq \widehat{q}_{\epsilon, -K} \} \widehat{q}_{\epsilon, -K} - \widehat{\psi}_l(\epsilon) + \epsilon \widehat{q}_{\epsilon; -K} ] / \widehat{\sigma}_l(\epsilon)) \\
	& \sup_{\epsilon \in \Ec} \sqrt{n} \Pn(\xi [\widehat{\psi}_u(\epsilon) - \epsilon \widehat{q}_{1-\epsilon; -K} - \varphi_u(\Ob; \widehat{\etab}_{-K}, \widehat{q}_{1-\epsilon; - K}) - \one\{g(\widehat{\etab}_{-K}) > \widehat{q}_{1-\epsilon, -K} \} \widehat{q}_{1-\epsilon, -K} ] / \widehat{\sigma}_u(\epsilon))
	\end{align*}
	respectively, where $(\zeta_1, \ldots, \zeta_n)$ and $(\xi_1, \ldots, \xi_n)$ are iid Rademacher random variables independent of the sample. Then, under the same conditions of Theorem \ref{gp_convergence}, it holds that
	\begin{align*}
	\Pb \{ [\psi_l(\epsilon), \psi_u(\epsilon)] \subseteq \widehat{\text{CI}}(\epsilon; \widehat{c}_\alpha, \widehat{d}_\alpha), \text{ for all } \epsilon \in \Ec \} \geq 1 - \alpha + o(1)
	\end{align*}
\end{lemma}
\begin{proof}
	Together with an application of the Bonferroni correction, the proof of Theorem 4 in \cite{kennedy2018nonparametric} can be used here.
\end{proof}
\subsection{Proof of Theorem 3} \label{esp_tilde_proof}
Recall the following map used to define $\epsilon_0$:
\begin{align*}
\Psi(\epsilon) = \psi_l(\epsilon)\psi_u(\epsilon) = \Pb\left\{ \varphi_l(\Ob; \etab; q_\epsilon) \right\}\Pb\left\{ \varphi_u(\Ob; \etab; q_{1-\epsilon}) \vphantom{q_{1-\epsilon}} \right\}
\end{align*}
where
\begin{align*}
\varphi_l(\Ob; \etab, q_\epsilon) = \nu(\Ob; \etab) + \tau(\Ob; \etab) \kappa_\epsilon -\epsilon,  \quad \varphi_u(\Ob; \etab, q_{1-\epsilon}) = \nu(\Ob; \etab) + \tau(\Ob; \etab) \lambda_{1-\epsilon},
\end{align*}
$\kappa_\epsilon = \one\left\{ g(\etab) \leq q_\epsilon \right\}$ and $\lambda_{1-\epsilon} = \one\left\{ g(\etab) > q_{1-\epsilon} \right\}$.
The corresponding empirical version, which makes use of cross-fitting, is:
\begin{align*}
\widehat{\Psi}_n(\epsilon) = \frac{1}{B}\sum_{k=1}^B \Pn^k\left\{ \widehat{\varphi}_l(\Ob; \widehat{\etab}_{-k}, \widehat{q}_{-k, \epsilon})  \right\}\Pn^k\left\{ \widehat{\varphi}_u(\Ob; \widehat{\etab}_{-k}, \widehat{q}_{-k, 1-\epsilon}) \right\}
\end{align*}
where $\Pn^k$ is the empirical measure over fold $k$, defined as in Section \ref{section:estimators}. 

The moment condition defining $\epsilon_0$ is $\Psi(\epsilon_0) = 0$, since at $\epsilon = \epsilon_0$ either the lower bound or the upper bound is equal to 0 and both are uniformly bounded so that the product is 0. Furthermore, the lower and upper bound curves are monotone in $\epsilon$; if the bounds are continuous and strictly monotone in a neighborhood of $\epsilon_0$, then the moment condition will be satisfied by a unique value in $[0, 1]$. In practice, we would estimate $\epsilon_0$ by $\epsilon_n$ solving the empirical moment condition $\widehat{\Psi}_n(\epsilon_n) = o_\Pb(n^{-1/2})$. 

Theorem \ref{eps_tilde_theorem} follows from a direct application of Theorem 3.3.1 in \cite{van1996weak}. Therefore, our proof consists of checking that the following conditions hold:
\begin{enumerate}
	\item $\sqrt{n}(\widehat{\Psi}_n - \Psi)(\epsilon_0) \indist N(0, \var\{\tilde{\varphi}(\Ob; \etab, \epsilon_0)\})$, where
	\begin{align*}
	\tilde{\varphi}(\Ob; \etab, \epsilon) = \psi_u(\epsilon)[\nu(\Ob; \etab) + \kappa_\epsilon \{\tau(\Ob; \etab) - q_\epsilon\} - \epsilon] + \psi_l(\epsilon)[\nu(\Ob; \etab) + \lambda_{1-\epsilon} \{\tau(\Ob; \etab) - q_{1-\epsilon} \}]
	\end{align*}
	\item $\sqrt{n}(\widehat{\Psi}_n - \Psi)(\epsilon_n) - \sqrt{n}(\widehat{\Psi}_n - \Psi)(\epsilon_0) = o_\Pb\left(1 + \sqrt{n} \left| \epsilon_n - \epsilon_0 \right| \right)$
	\item The map $\epsilon \mapsto \Psi(\epsilon)$ is differentiable at $\epsilon = \epsilon_0$. 
	\item $\epsilon_n$ is such that $\widehat{\Psi}_n(\epsilon_n) = o_\Pb(n^{-1/2})$ and $\epsilon_n \inprob \epsilon_0$. 
\end{enumerate}
We will follow the same notation as for the proof of Theorem \ref{gp_convergence}. In particular, let $\|  f \| _\Ec \ = \sup_{\epsilon \in \Ec} |  f(\epsilon) |  $ denote the supremum norm over $\Ec$. We proceed with considering $\Ec = [0, 1]$. 
\subsubsection{Proof of Statement 1}
We actually prove the following stronger result:
\begin{align*}
\| \sqrt{n} (\widehat{\Psi}_n - \Psi) - \sqrt{n}(\Pn - \Pb) \tilde{\varphi} \|_\Ec = o_\Pb(1),
\end{align*}
for $\tilde{\varphi}(\cdot; \etab, \epsilon)$ living in a Donsker class. This is useful in establishing the other conditions. 

First, we claim that the function $\tilde{\varphi}(\cdot; \etab, \epsilon)$ lives in a Dosker class. To see this, notice that 
\begin{align*}
\tilde{\varphi}(\Ob; \etab, \epsilon) = \psi_u(\epsilon)\overline{\varphi}_l(\Ob; \etab, \epsilon) + \psi_l(\epsilon)\overline{\varphi}_u(\Ob; \etab, \epsilon)
\end{align*}
where $\overline{\varphi}_l(\Ob; \etab, \epsilon) = \nu(\Ob; \etab) + \kappa_\epsilon \{\tau(\Ob; \etab) - q_\epsilon\} - \epsilon$ and $\overline{\varphi}_u(\Ob; \etab, \epsilon) = \nu(\Ob; \etab) + \lambda_{1-\epsilon} \{\tau(\Ob; \etab) - q_{1-\epsilon} \}$.
In the proof of Theorem \ref{gp_convergence}, we showed that $\overline{\varphi}_u(\cdot; \etab, \epsilon)$ lives in a Donsker class because its class can be constructed via sums and products of classes of uniformly bounded, monotone functions. Therefore, following a similar logic, we conclude that $\tilde{\varphi}(\cdot; \etab, \epsilon)$ lives in a Donsker class as well. 

Next, we argue that $\| \sqrt{n} (\widehat{\Psi}_n - \Psi) - \sqrt{n}(\Pn - \Pb) \tilde{\varphi} \|_\Ec = o_\Pb(1)$. A bit of algebra reveals that
\begin{align*}
\widehat{\Psi}_n(\epsilon) - \Psi(\epsilon) - (\Pn - \Pb) \tilde{\varphi} & = \frac{1}{B} \sum_{k=1}^B \left[ \vphantom{\widehat{\lambda}_{-k, 1-\epsilon}}(\Pn^k - \Pb)(\widehat{\overline{\varphi}}_{l, -k} - \overline{\varphi}_l)(\Pn^k - \Pb)(\widehat{\overline{\varphi}}_{u, -k} - \overline{\varphi}_u) \right. \\
& \hphantom{= \frac{1}{B} \sum_{k=1}^B \left\{\vphantom{\widehat{\lambda}_{-k, 1-\epsilon}}\right.} + ( \Pn^k - \Pb)(\widehat{\overline{\varphi}}_{l, -k} - \overline{\varphi}_l) \{T_1 - T_2 + (\Pn - \Pb)(\varphi_{u})\} \\
& \hphantom{= \frac{1}{B} \sum_{k=1}^B \left\{\vphantom{\widehat{\lambda}_{-k, 1-\epsilon}}\right.} + (\Pn^k - \Pb)(\widehat{\overline{\varphi}}_{u, -k} - \overline{\varphi}_u)\{V_1 - V_2 + (\Pn - \Pb)(\varphi_l)\} \\
& \hphantom{= \frac{1}{B} \sum_{k=1}^B \left\{\vphantom{\widehat{\lambda}_{-k, 1-\epsilon}}\right.} + ( \Pn - \Pb)(\varphi_u)(V_1 - V_2) + (\Pn - \Pb)(\varphi_l)(T_1 - T_2) \\
& \hphantom{= \frac{1}{B} \sum_{k=1}^B \left\{\vphantom{\widehat{\lambda}_{-k, 1-\epsilon}}\right.} + (\Pn - \Pb)(\varphi_u)(\Pn - \Pb)(\varphi_l) \\
& \hphantom{= \frac{1}{B} \sum_{k=1}^B \left\{\vphantom{\widehat{\lambda}_{-k, 1-\epsilon}}\right.} + T_1V_1 - T_1V_2 - T_2V_1 + T_2V_2  \\
& \hphantom{= \frac{1}{B} \sum_{k=1}^B \left\{\vphantom{\widehat{\lambda}_{-k, 1-\epsilon}}\right.} + \Pb(\varphi_u) \{(\Pn^k - \Pb)(\widehat{\overline{\varphi}}_{l, -k} - \overline{\varphi}_l) + V_1\} \\
& \hphantom{= \frac{1}{B} \sum_{k=1}^B \left\{\vphantom{\widehat{\lambda}_{-k, 1-\epsilon}}\right.} + \Pb(\varphi_l) \{(\Pn^k - \Pb)(\widehat{\overline{\varphi}}_{u, -k} - \overline{\varphi}_u) + T_1 \} \left. \right]
\end{align*} 
where
\begin{align*}
& \widehat{\overline{\varphi}}_{l, -k} - \overline{\varphi}_l = \widehat{\varphi}_{l -k} - q_{\epsilon} \widehat{\kappa}_{-k, \epsilon} - \varphi_{l -k} + q_{\epsilon}\kappa_{\epsilon} \\
& \widehat{\overline{\varphi}}_{u, -k} - \overline{\varphi}_u = \widehat{\varphi}_{u-k} - q_{1-\epsilon} \widehat{\lambda}_{-k, 1-\epsilon} - \varphi_{u} + q_{1-\epsilon}\lambda_{1-\epsilon} \\
& V_1 = \Pb\{ \widehat{\nu}_{-k} - \nu + \widehat{\kappa}_{-k, \epsilon}(\widehat{\tau}_{-k} - \tau) + (\tau - q_{\epsilon})(\widehat{\kappa}_{-k, \epsilon} - \kappa_\epsilon)\} \\
& T_1 = \Pb\{ \widehat{\nu}_{-k} - \nu + \widehat{\lambda}_{-k, 1-\epsilon}(\widehat{\tau}_{-k} - \tau) + (\tau - q_{1-\epsilon})( \widehat{\lambda}_{-k, 1-\epsilon} - \lambda_{1-\epsilon})\} \\
& V_2 = (\Pn - \Pb)( q_{\epsilon}\kappa_{\epsilon}) \quad \text{ and } \quad T_2 = (\Pn - \Pb)( q_{1-\epsilon}\lambda_{1-\epsilon})
\end{align*}
As shown in the proof of Theorem \ref{gp_convergence}, under the conditions of the theorem, it holds that
\begin{align*}
& \left\| \frac{1}{B}\sum_{k=1}^B (\Pn^k - \Pb) (\widehat{\overline{\varphi}}_{l, -k} - \overline{\varphi}_l) \right\|_\Ec = o_\Pb(n^{-1/2}), \quad \left\|\frac{1}{B}\sum_{k=1}^B (\Pn^k - \Pb) (\widehat{\overline{\varphi}}_{u, -k} - \overline{\varphi}_u) \right\|_\Ec = o_\Pb(n^{-1/2}), \\
& \| V_1 \|_\Ec = o_\Pb(n^{-1/2}), \quad \| T_1 \|_\Ec = o_\Pb(n^{-1/2}), \quad  \| V_2 \|_\Ec = O_\Pb(n^{-1/2}), \quad \text{and} \quad \| T_2 \|_\Ec = O_\Pb(n^{-1/2}).
\end{align*}
Therefore, by an application of the triangle inequality, it holds that
\begin{align*}
\| \sqrt{n}(\widehat{\Psi}_n - \Psi) - \sqrt{n}(\Pn - \Pb) \tilde{\varphi} \|_\Ec = o_\Pb(1)
\end{align*}
In particular, 
\begin{align*}
\sqrt{n} (\widehat{\Psi}_n - \Psi)(\epsilon_0) \indist N(0, \var\{\tilde{\varphi}(\Ob; \etab, \epsilon_0)\})
\end{align*}
by Slutsky's theorem.
\subsubsection{Proof of Statement 2}
Because in the proof of Statement 1 we have argued that
\begin{align*}
\| \sqrt{n} (\widehat{\Psi}_n - \Psi) - \sqrt{n}(\Pn - \Pb)\tilde{\varphi}\|_\Ec = \ o_\Pb(1)
\end{align*}
to prove Statement 2, it is sufficient to show
\begin{align}\label{statement_2}
\sqrt{n} (\Pn - \Pb) \{\tilde{\varphi}(\Ob; \etab, \epsilon_n) \} - \sqrt{n}(\Pn - \Pb) \{\tilde{\varphi}(\Ob; \etab, \epsilon_0) \} = o_\Pb(1 + \sqrt{n} \left| \epsilon_n - \epsilon_0 \right|)
\end{align}
Because $\tilde{\varphi}(\cdot; \etab, \epsilon)$ lives in a Donsker class and $\epsilon_n \inprob \epsilon_0$ (proved below in the proof of Statement 4), by Lemma 3.3.5 in \cite{van1996weak}, in order to prove \eqref{statement_2} it is sufficient to show that
\begin{align*}
\Pb \{\tilde{\varphi}(\epsilon) - \tilde{\varphi}(\epsilon_0)\}^2 \to 0 \quad \text{ as } \quad \epsilon \to \epsilon_0
\end{align*}
We have that
\begin{align*}
\Pb\{\tilde{\varphi}(\epsilon) - \tilde{\varphi}(\epsilon_0)\}^2 & = \Pb[ \psi_l(\epsilon)\{\nu + \lambda_{1-\epsilon}( \tau - q_{1-{\epsilon}} )\} - \psi_l(\epsilon_0)\{\nu + \lambda_{1-\epsilon_0}(\tau - q_{1-{\epsilon_0}})\} \\
& \hphantom{= \Pb[}+ \psi_u(\epsilon)\{\nu + \kappa_{\epsilon} ( \tau - q_{\epsilon}) - \epsilon\} - \psi_u(\epsilon_0)\{\nu + \kappa_{\epsilon_0}(\tau - q_{\epsilon_0}) - \epsilon_0\}]^2 \\
& = \Pb \left(D_l + D_u\right)^2
\end{align*}
Notice that we can write
\begin{align*}
D_l & = \{\psi_l(\epsilon) - \psi_l(\epsilon_0)\}\{\nu + \lambda_{1-\epsilon}(\tau - q_{1-{\epsilon}})\} \\
& \hphantom{D_l = } + \psi_l(\epsilon_0) \{(\lambda_{1-\epsilon} - \lambda_{1-\epsilon_0})( \tau - q_{1-{\epsilon}})  + \lambda_{1-\epsilon_0}(q_{1-{\epsilon_0}} - q_{1-{\epsilon}})\} \\
D_u & = \{\psi_u(\epsilon) - \psi_u(\epsilon_0)\}(\nu + \kappa_{\epsilon}(\tau - q_{\epsilon}) - \epsilon_n) \\
& \hphantom{D_u = } + \psi_u(\epsilon_0)\{(\kappa_{\epsilon} - \kappa_{\epsilon_0})(\tau - q_{\epsilon}) + \kappa_{\epsilon_0}( q_{\epsilon_0} - q_{\epsilon}) - (\epsilon - \epsilon_0)\} 
\end{align*}
Then, we have
\begin{align*}
& \Pb(D_l^2) \lesssim \Pb|\lambda_{1-\epsilon} - \lambda_{1-\epsilon_0}| \ + \ |q_{1-\epsilon_0} - q_{1-\epsilon}| \ + \ |\psi_l(\epsilon) - \psi_l(\epsilon_0)| \\
& \Pb(D_u^2) \lesssim \Pb|\kappa_{\epsilon} - \kappa_{\epsilon_0}| \ + \ |q_{\epsilon_0} - q_{\epsilon}| \ + \ |\psi_u(\epsilon) - \psi_u(\epsilon_0)| \ + \ |\epsilon - \epsilon_0| \\
& \Pb(D_lD_u) \lesssim \ |\psi_l(\epsilon) - \psi_l(\epsilon_0)| \ + \ |\psi_u(\epsilon) - \psi_u(\epsilon_0)| \ + \ \Pb|\lambda_{1-\epsilon} - \lambda_{1-\epsilon_0}| \ + \ \Pb|\kappa_{\epsilon} - \kappa_{\epsilon_0}|\\
& \hphantom{\Pb \left(D_lD_u\right)  \lesssim}  + |q_{1-\epsilon_0} - q_{1-\epsilon}| \ + \ |q_{\epsilon_0} - q_{\epsilon}| \ + \ |\epsilon - \epsilon_0|
\end{align*}
Next, notice
\begin{align*}
& \Pb| \kappa_{\epsilon} - \kappa_{\epsilon_0}| \ \leq \Pb[\one\{| g(\etab) - q_{\epsilon_0}| \ \leq | q_{\epsilon_0} - q_{\epsilon} |\}] \lesssim | q_{\epsilon_0} - q_{\epsilon}|^\alpha \\
& \Pb|\lambda_{1-\epsilon} - \lambda_{1-\epsilon_0}| \ \leq \Pb[\one\{|g(\etab) - q_{1-\epsilon_0}| \ \leq | q_{1-\epsilon_0} - q_{1-\epsilon} |\}] \lesssim | q_{1-\epsilon_0} - q_{1-\epsilon}|^\alpha
\end{align*}
for some $\alpha > 0$. The first inequalities rely on Lemma \ref{ind_lemma}. The last step hinges on the fact that the density of $g(\etab)$ satisfies the margin condition 3 for some $\alpha > 0$. 

Moreover, we have
\begin{align*}
| \psi_l(\epsilon) - \psi_l(\epsilon_0)| \ \lesssim \Pb|\kappa_{\epsilon} - \kappa_{\epsilon_0}| \ + \ |\epsilon - \epsilon_0|  \quad \text{ and } \quad | \psi_u(\epsilon) - \psi_u(\epsilon_0) | \ \lesssim \Pb|\lambda_{1-\epsilon} - \lambda_{1-\epsilon_0}|
\end{align*}
since $\Pb(|g(\etab)| \ \leq 1) = 1$. 

We have assumed that the CDF of $g(\etab)$ is continuous and strictly increasing in neighborhoods of $q_{\epsilon_0}$ and $q_{1 - \epsilon_0}$, thus the quantile function is continuous in neighborhoods of $\epsilon_0$ and $1 - \epsilon_0$ as well, allowing us to conclude that, for $\alpha > 0$
\begin{align*}
| q_{\epsilon_0} - q_{\epsilon}|^\alpha \to 0 \quad \text{ and } \quad |q_{1-\epsilon_0} - q_{1-\epsilon}|^\alpha \to 0 \quad \text{ as } \quad \epsilon \to \epsilon_0.
\end{align*} 
Then, it follows that $\Pb \{\tilde{\varphi}(\epsilon) - \tilde{\varphi}(\epsilon_0)\}^2 \to 0$ as $\epsilon \to \epsilon_0$. 

\subsubsection{Proof of Statement 3}
To prove Statement 3, notice that
\begin{align*}
\psi_l(\epsilon)\psi_u(\epsilon) = \left[ \E \left\{\mu_1(\Xb) - \mu_0(\Xb)\right\} + \int_{0}^{q_\epsilon} t dG(t) - \epsilon \right] \left[ \E \left\{\mu_1(\Xb) - \mu_0(\Xb)\right\} + \int_{q_{1-\epsilon}}^{1} t dG(t) \right]
\end{align*}
Because we have assumed that the quantile function of $g(\etab)$ is differentiable in neighborhoods of $\epsilon_0$ and $1-\epsilon_0$, by ``Leibniz integral rule," it holds that
\begin{align*}
\Psi^{'}(\epsilon_0) = \left. \frac{d}{d \epsilon} \psi_l(\epsilon)\psi_u(\epsilon) \right|_{\epsilon = \epsilon_0} = \psi_u(\epsilon_0) (q_{\epsilon_0} - 1) + \psi_l(\epsilon_0)q_{1-\epsilon_0}
\end{align*}
which we have assumed to be nonzero. Notice that in calculating the derivative, we used the fact that $\int t dG(t) = \int t f(t) dt$ with $f$ being the density of $g(\etab)$, which we have assumed to exist. 
\subsubsection{Proof of Statement 4}
We have that $\Psi_n(\epsilon_n) = o_\Pb(n^{-1/2})$ by definition. Furthermore, we have shown that
\begin{align*}
\| \Psi_n - \Psi \|_\Ec \ =  \| (\Pn - \Pb)\{ \tilde{\varphi}(\Ob; \etab, \epsilon) \}\|_\Ec \ + \ o_\Pb(n^{-1/2}) = o_\Pb(1)
\end{align*}
where the last equality follows because $\tilde{\varphi}(\cdot; \etab, \epsilon)$ is Donsker and thus Glivenko-Cantelli.

We now show that $\psi_l(\epsilon)$ and $\psi_u(\epsilon)$ are strictly monotone. First, for $\epsilon_1 < \epsilon_2$, we have
\begin{align*}
\psi_l(\epsilon_1) - \psi_l(\epsilon_2) & = \E(g(\etab)[\one\{g(\etab) \leq q_{\epsilon_1} \} - \one\{g(\etab) \leq q_{\epsilon_2} \}]) - (\epsilon_1 - \epsilon_2) \\
& = - \E\{g(\etab) \mid q_{\epsilon_1} < g(\etab) < q_{\epsilon_2}\}\Pb\{q_{\epsilon_1} < g(\etab) < q_{\epsilon_2}\}  - (\epsilon_1 - \epsilon_2) \\
& = - \E\{g(\etab) \mid q_{\epsilon_1} < g(\etab) < q_{\epsilon_2}\}(\epsilon_2 - \epsilon_1)  + (\epsilon_2 - \epsilon_1) \\
& > 0
\end{align*}
where we used the facts that $\Pb\{0 < g(\etab) < 1\} = 1$ and $\Pb\{q_{\epsilon_1} < g(\etab) < q_{\epsilon_2}\} = \epsilon_2 - \epsilon_1$ (continuity of $g(\etab)$), that $\one\{g(\etab) \leq q_{\epsilon_1} \} \leq \one\{g(\etab) \leq q_{\epsilon_2} \}$ (monotonicity of quantile function) and that
\begin{align*}
\one\{g(\etab) \leq q_{\epsilon_1} \} - \one\{g(\etab) \leq q_{\epsilon_2} \} = - 1 \iff q_{\epsilon_1} < g(\etab) < q_{\epsilon_2}
\end{align*}
Similarly, we note that, for $\epsilon_1 < \epsilon_2$, we have
\begin{align*}
\psi_u(\epsilon_1) - \psi_u(\epsilon_2) & = \E(g(\etab)[\one\{g(\etab) > q_{1 - \epsilon_1} \} - \one\{g(\etab) > q_{1 - \epsilon_2} \}]) \\
& = - \E\{g(\etab) \mid q_{1 - \epsilon_2} < g(\etab) < q_{1 - \epsilon_1}\}(\epsilon_2 - \epsilon_1) \\
& < 0
\end{align*}
using the same logic as before. Thus, we conclude that, under the assumption that $g(\etab)$ is a continuous random variable, both $\psi_l(\epsilon)$ and $\psi_u(\epsilon)$ are continuous and strictly monotone. Therefore, the value $\epsilon_0$ satisfying $\Psi(\epsilon_0) = 0$ must be unique. Furthermore, we have assumed (to derive a finite asymptotic variance of $\epsilon_n$) that $\Psi^{'}(\epsilon_0) \neq 0$, thus a first-order Taylor expansion of $\Psi(\epsilon_n)$ around $\epsilon_0$
\begin{align*}
\Psi(\epsilon_n) = \Psi^{'}(\epsilon_0)(\epsilon_n - \epsilon_0) + o(|\epsilon_n - \epsilon_0|)
\end{align*}
suffices to conclude that $|\Psi(\epsilon_n)| \ \to 0$ implies $|\epsilon_n - \epsilon_0| \ \to 0$ for any sequence $\epsilon_n \in \Ec$. In other words, under the assumptions of the theorem, the identifiability condition of $\epsilon_0$ is satisfied. Then, by an application of Theorem 2.10 in \cite{kosorok2008introduction}, we conclude that $|\epsilon_n - \epsilon_0| = o_\Pb(1)$ as desired. 
\section{Additional Data Analysis}
In this section, we provide additional analysis of the data from \cite{connors1996effectiveness}. In Figure \ref{connors_results_2D}, we consider values of $\delta$ smaller than 1, and notice that the bounds would start to include zero for larger values of $\epsilon$. For instance, under the $X$-mixture model, if $\delta = 1/2$ is used, the results appear to be robust for up to $11.00\%$ ($95\% \text{ CI}=[3.84\%, 18.16\%]$) of confounded units in the sample. A value of $\delta = 1/2$ requires that the counterfactual mean outcomes satisfy:
\begin{align*}
\frac{\mu_a(\Xb)}{2} \leq \E(Y^a \mid A = 1-a, \Xb, S = 0) \leq \frac{1}{2} + \frac{\mu_a(\Xb)}{2} \text{ with prob. 1.}
\end{align*}
for $a \in\{0, 1\}$, thereby restricting $\E(Y^a \mid A = 1-a, \Xb, S = 0)$ to be in an interval of length 1/2 instead of the worst-case interval of length 1. Robustness is up to $8.99\%$ ($95\% \text{ CI}=[3.78\%, 14.20\%]$) confounded units if the $XA$-mixture model is considered instead. 
\begin{figure}[H]
	\subfloat{%
		\includegraphics[width=0.5\textwidth]{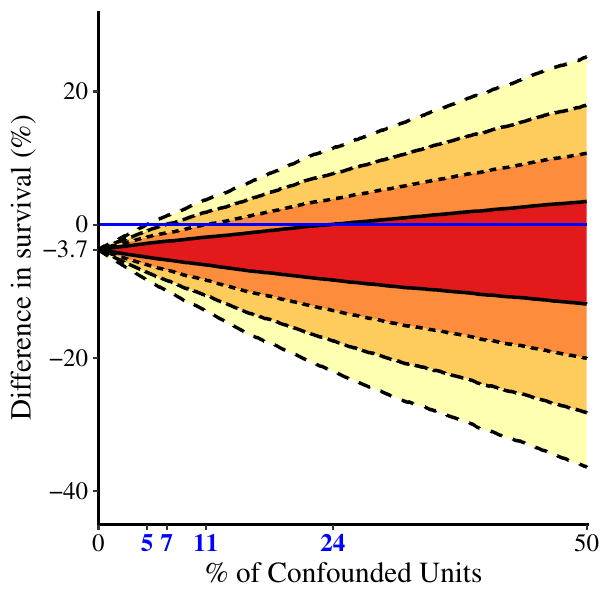}}
	\hspace{\fill}
	\subfloat{%
		\includegraphics[width=0.5\textwidth]{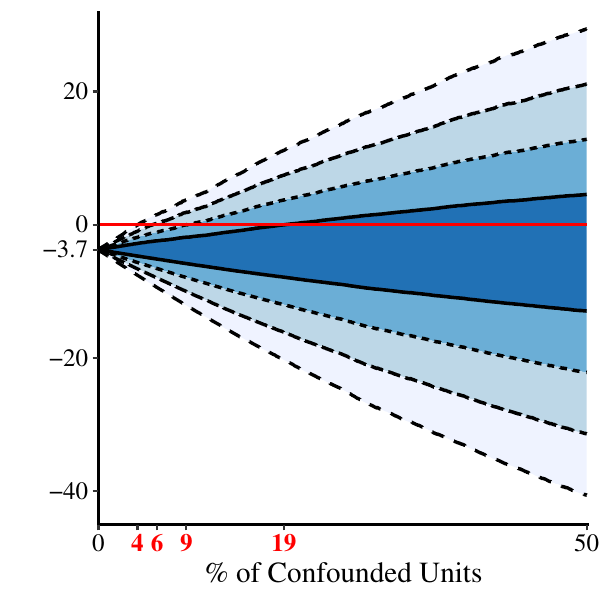}}\\
	\caption{\label{connors_results_2D} Estimated bounds on the average treatment effect as a function of the proportion of confounded units $\epsilon$ and the parameter $\delta \in \{0.25, 0.5, 0.75, 1\}$, which governs the amount of confounding among the $S = 0$ units. Darker shades correspond to smaller values of $\delta$. Bolded labels on the abscissa represent estimates of $\epsilon_0$ for corresponding values of $\delta$. Uniform and pointwise confidence intervals are not shown for the sake of clarity.}
\end{figure}

\subsection{Results using the sensitivity model from \cite{cinelli2020making}} \label{section:extra_data_analysis}
In this section, we briefly report the results from applying the sensitivity analysis for linear models discussed in \cite{cinelli2020making}. Because their model is appropriate only for causal effect estimates computed using OLS, we fit a linear model for 30-day survival regressed all baseline covariates, the treatment and no interactions. If the model is accurate and there is no residual confounding, RHC usage appears to decrease the probability of 30-day survival by 0.042 ($95\% \text{CI} = [-0.07, -0.02]$). However, it is sufficient that an unmeasured confounder explains 1.7\% of the outcome variance not already captured by the treatment and the covariates and 1.7\% of the treatment variance not already captured by the covariates to make the effect not statistically significance at the $0.05$-level (Table \ref{table:cinelli_hazlett_report}). 
\begin{table}[!h]
	\centering
	\begin{tabular}{lrrrrrr}
		\multicolumn{7}{c}{Outcome: survival at day 30} \\
		\hline \hline 
		Treatment: & Est. & S.E. & t-value & $R^2_{Y \sim D |{\bf X}}$ & $RV_{q = 1}$ & $RV_{q = 1, \alpha = 0.05}$  \\ 
		\hline 
		RHC usage & -0.042 & 0.013 & -3.261 & 0.2\% & 4.2\% & 1.7\% \\ 
		\hline 
		df = 5658 & & \multicolumn{5}{r}{ \small \textit{Bound (2x \texttt{dnr1})}: $R^2_{Y\sim Z| {\bf X}, D}$ = 3.7\%, $R^2_{D\sim Z| {\bf X} }$ = 1\%} \\
		df = 5658 & & \multicolumn{5}{r}{ \small \textit{Bound (2x \texttt{is\_miss\_adld3p})}: $R^2_{Y\sim Z| {\bf X}, D}$ = 6\%, $R^2_{D\sim Z| {\bf X} }$ = 1.6\%} \\
		\hline
	\end{tabular}
	\caption{\label{table:cinelli_hazlett_report} Summary of the effect estimate under no unmeasured confounding as well as assessment of the estimate's robustness using some key covariates as benchmarks.}
\end{table}
As shown in Figure \ref{figure:cinelli_hazlett_plots}, the observed effect would cease to be significance also if there is an unmeasured confounder with explanatory power that is 2 times greater than that of the variable \texttt{dnr1} (an indicator for whether there was a ``do not resuscitate" order when the patient was admitted on day 1) or 2 times greater than that of the variable indicating that \texttt{adld3p} (ADL) is missing.  

\begin{figure}[H]
	\subfloat{%
		\includegraphics[width=0.5\textwidth]{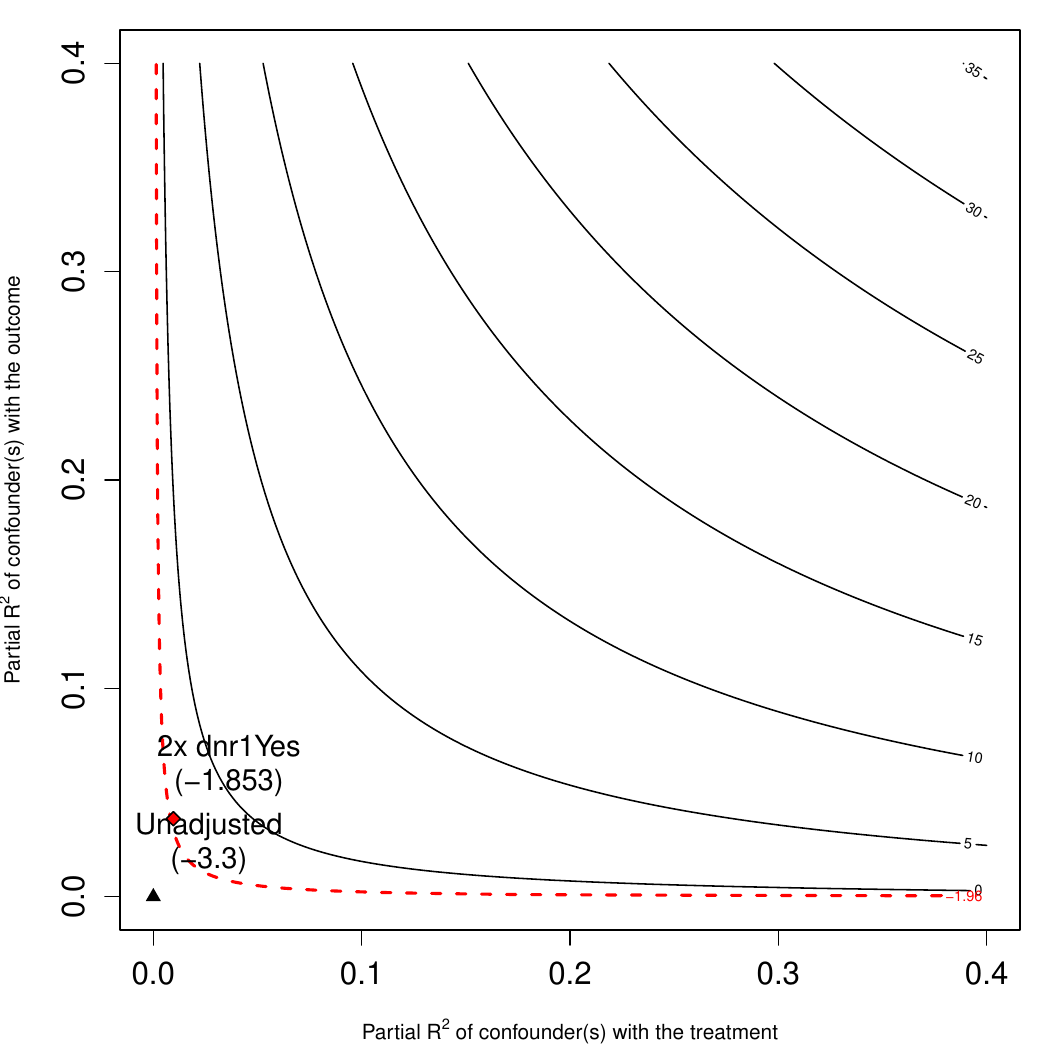}}
	\hspace{\fill}
	\subfloat{%
		\includegraphics[width=0.5\textwidth]{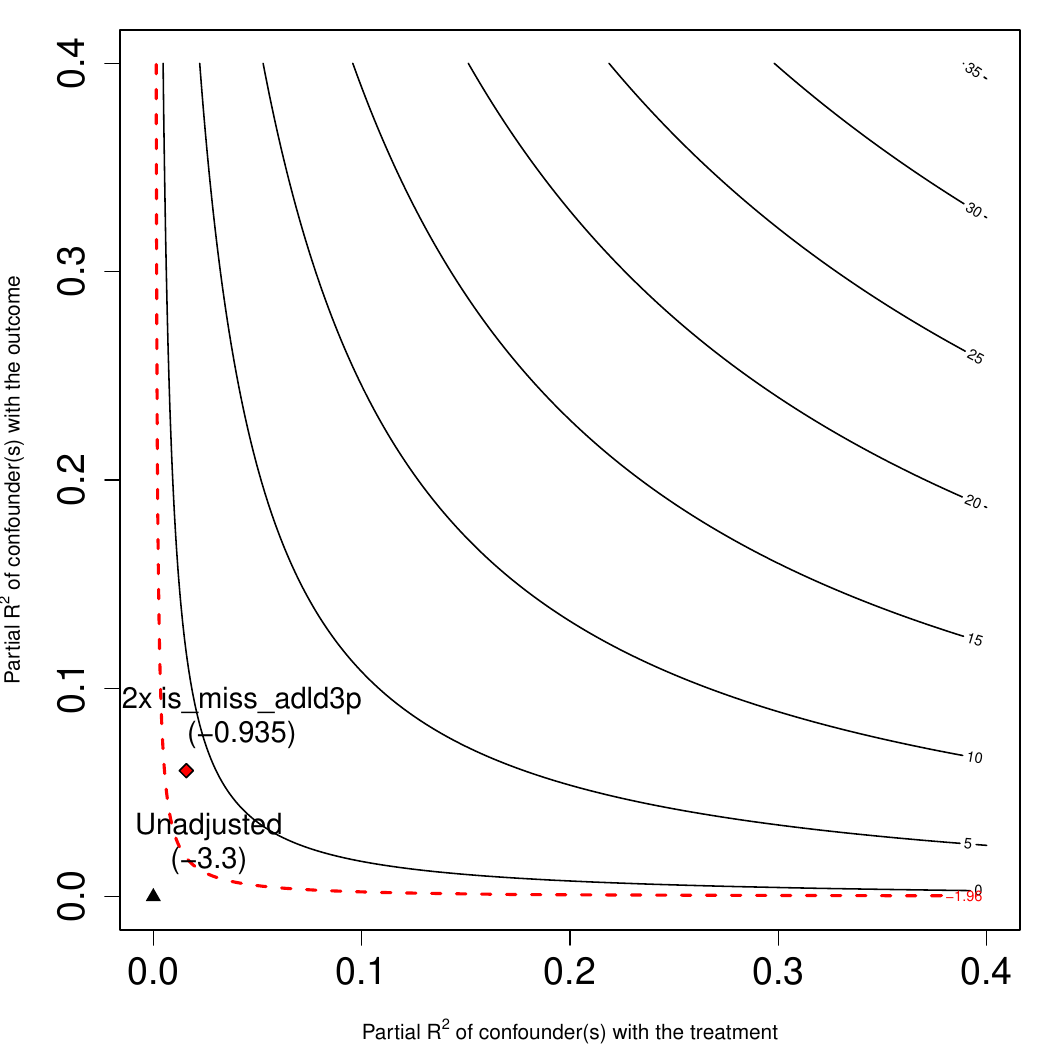}}
	\caption{\label{figure:cinelli_hazlett_plots} Sensitivity contour plots in the partial $R^2$ scale with benchmark bounds of the $t$-value.}
\end{figure}

\section{Simulations regarding power}\label{sec:power_sensitivity}
In this section, we conduct a brief simulation to investigate how conservative inference based on $\epsilon_0$ is when all the confounders have been measured so that the true $\epsilon$ is actually zero. The bounds on the ATE $\tau$ depends on three fundamental quantities, $\mu_a(\Xb) = \E(Y \mid A = a, \Xb)$ for $a = 0, 1$ and $\pi(\Xb) = \Pb(A = 1 \mid \Xb)$. Let  $Q_g(p, \delta)$ be the quantile function of either $g(\etab)$ as defined in Theorem 1 ($X$-model) or $g(A, \etab)$ as defined in Section \ref{appendix:xa-model} ($XA$-model). When $\epsilon = 0$, $\tau = \E\{\mu_1(\Xb) - \mu_0(\Xb)\}$ and the bounds can be written as
\begin{align*}
\psi_l(\epsilon, \delta) = \tau + \int_0^\epsilon Q_g(p, \delta) dp - \epsilon \delta(y_\max - y_\min) \quad { \text{ and }} \quad 	\psi_u(\epsilon, \delta) = \tau + \int_{1 - \epsilon}^1 Q_g(p, \delta) dp.
\end{align*}
Therefore, we may define the design sensitivity \citep{rosenbaum2004design} as $\widetilde \epsilon$ solving
\begin{align*}
& \tau + \int_0^{\widetilde\epsilon} Q_g(p, \delta) dp - \widetilde \epsilon \delta(y_\max - y_\min) = 0 \text{ if } \tau \geq 0 \quad { \text{ and }} \quad \tau + \int_{1 - \widetilde\epsilon}^1 Q_g(p, \delta) dp  = 0 \text{ if } \tau < 0.
\end{align*}
Thus, $\widetilde\epsilon$ depends on $\tau$ and the quantile function $Q_g(p, \delta)$, which itself depends on $\tau$ through the functions $\mu_a(\Xb)$. Without knowing $Q_g(p, \delta)$, one can get crude bounds on $\widetilde\epsilon$ as
\begin{align*}
& \frac{\tau}{\delta(y_\max - y_\min) - Q_g(0, \delta)} \leq \widetilde\epsilon \leq \frac{\tau}{\delta(y_\max - y_\min) - Q_g(1, \delta)} & \text{ if } \tau \geq 0 \\
& \frac{|\tau|}{ Q_g(1, \delta)} \leq \widetilde\epsilon \leq \frac{|\tau|}{Q_g(0, \delta)} & \text{ if } \tau < 0
\end{align*} 
since, for example, $\int_0^{\epsilon} Q_g(p, \delta) dp \geq \epsilon Q_g(0, \delta)$. The derivatives of the bounds are
\begin{align*}
\frac{d}{d\epsilon}\psi_l(\epsilon, \delta) =  Q_g(\epsilon, \delta) - \delta(y_\max - y_\min) \quad \text{and} \quad \frac{d}{d\epsilon}\psi_u(\epsilon, \delta) =  Q_g(1 - \epsilon, \delta).
\end{align*}
so that the rate at which they widen crucially depends on $Q_g(p, \delta)$. For example, let the data be generated as in the simulation setup of Section \ref{section:sim} except for 
\begin{align*}
Y^{a} \mid X_1, X_2, U, S, A \sim \text{Bern}\{(1 - a)/2 + a B^{-1}_{\alpha} \circ TN(X_1) \},
\end{align*}
where $B^{-1}_{\alpha}(\cdot)$ is the quantile function of a $\text{Beta}(\alpha, 1)$ random variable and $TN(\cdot)$ is the CDF of $X_1$, a truncated normal random variable in $[-2, 2]$. Therefore, $\mu_0(X_1, X_2) = 1/2$, $\mu_1(X_1, X_2) \sim \text{Beta}(\alpha, 1)$, and $\tau = \alpha / (\alpha + 1) - 1/2$. As shown in Figure \ref{fig:design_sensitivity}, as $\tau$ increases $\widetilde{\epsilon}$ increases, although the relationship is nonlinear. In fact, different values of $\alpha$ also affects the skewness of the distribution of $g(\etab)$, for example. Finally, $\widetilde\epsilon$ is always greater under the $X$-model than under the $XA$-model because the bounds under the latter are at least as wide as those under the former model.  
\begin{figure}[H]
	\centering
	\includegraphics[scale=0.25]{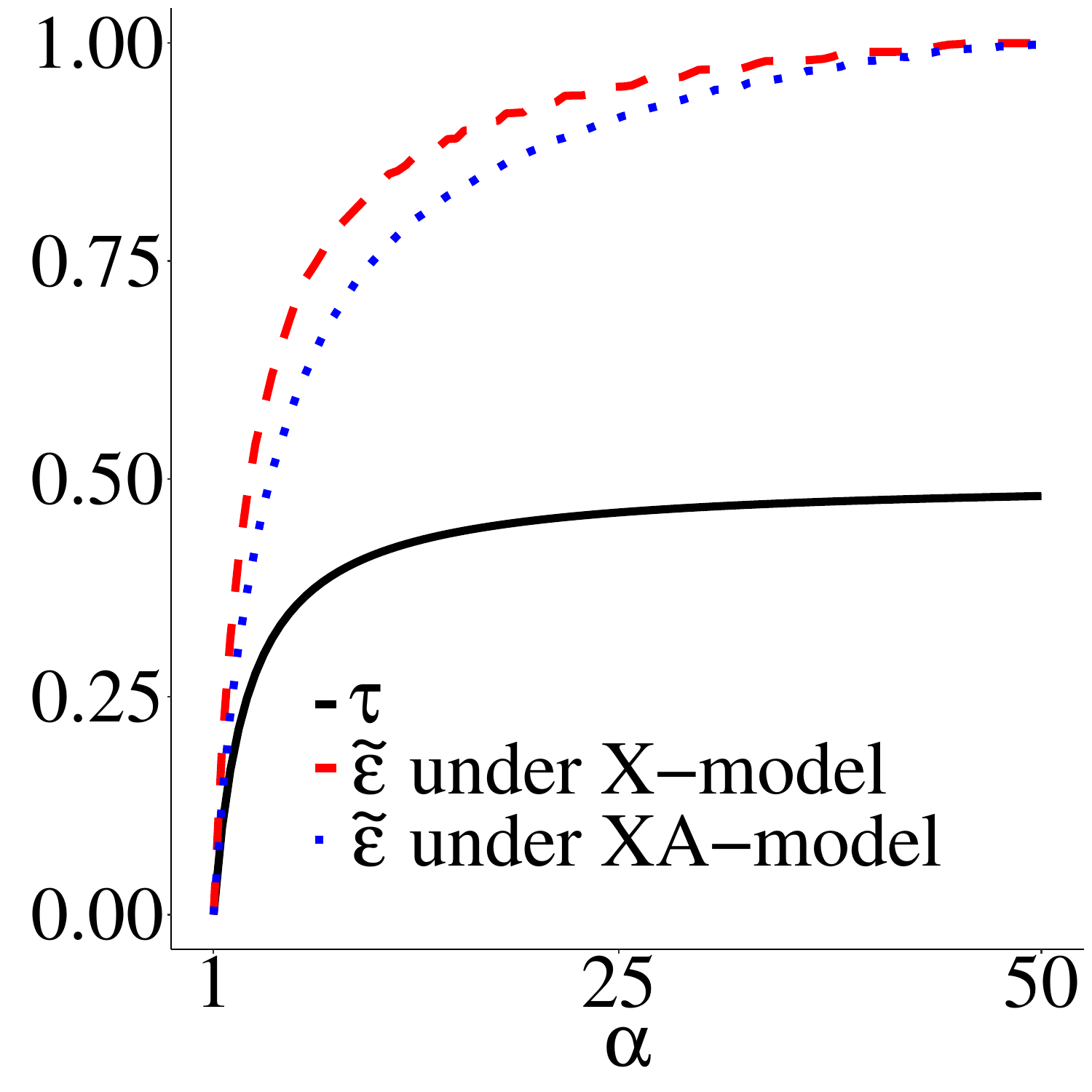}	
	\caption{\label{fig:design_sensitivity} The values of $\tau$ and $\widetilde{\epsilon}$ under either the $X$-model or the $XA$-model are shown as a function of $\alpha$. }
\end{figure}
\end{document}